\documentclass[reqno,11pt]{amsart}
\usepackage[utf8]{inputenc}
\usepackage{graphicx}
\usepackage{slashed}
\usepackage{amscd}
\usepackage{amssymb}
\usepackage{esint}
\usepackage{enumitem}
\usepackage[arrow, matrix, curve]{xy}
\usepackage[mathscr]{eucal}
\numberwithin{equation}{section}
\allowdisplaybreaks[1]

\title[Existence of Minimizers on Compact Subsets of Momentum Space]{Existence of Minimizers for Causal Variational Principles on Compact Subsets of Momentum Space in the Homogeneous Setting}

\author[C.\ Langer]{Christoph Langer \\ \\ August 2021}
\address{Fakult\"at f\"ur Mathematik \\ Universit\"at Regensburg \\ D-93040 Regensburg \\ Germany}
\email{
	christoph.langer@ur.de}

\newtheorem{Def}{Definition}[section]
\newtheorem{Thm}[Def]{Theorem}
\newtheorem{Prp}[Def]{Proposition}
\newtheorem{Lemma}[Def]{Lemma}
\newtheorem{Remark}[Def]{Remark}
\newtheorem{Corollary}[Def]{Corollary}

\newcommand{\Thanks}{\vspace*{.5em} \noindent \thanks}
\newcommand{\beq}{\begin{equation}}
\newcommand{\eeq}{\end{equation}}
\newcommand{\Proof}{\begin{proof}}
	\newcommand{\QED}{\end{proof} \noindent}

\newcommand{\C}{\mathbb{C}}
\newcommand{\R}{\mathbb{R}}
\newcommand{\Id}{\mbox{\rm 1 \hspace{-1.05 em} 1}}

\newcommand{\N}{\mathbb{N}}

\DeclareMathOperator{\tr}{tr}
\DeclareMathOperator{\Tr}{Tr}

\renewcommand{\L}{{\mathcal{L}}}
\newcommand{\LL}{{\text{\rm{L}}}}
\newcommand{\Sact}{{\mathcal{S}}}
\newcommand{\T}{{\mathcal{T}}}

\newcommand\B{{\mathscr{B}}}

\renewcommand{\H}{\mathscr{H}}

\DeclareMathOperator{\diag}{diag}

\newcommand{\F}{{\mathscr{F}}}

\DeclareMathOperator{\re}{Re}
\DeclareMathOperator{\im}{Im}

\DeclareMathOperator{\supp}{supp}

\newcommand{\scrM}{\mycal M}
\newcommand{\hscrM}{\,\,\hat{\!\!\scrM}}

\newcommand{\bitem}{\begin{itemize}[leftmargin=2.5em]}
	\newcommand{\eitem}{\end{itemize}}

\newcommand{\ovm}{{\mathfrak{Ovm}}}
\newcommand{\ndm}{{\mathfrak{Ndm}}}

\DeclareFontFamily{OT1}{rsfso}{}
\DeclareFontShape{OT1}{rsfso}{m}{n}{ <-7> rsfso5 <7-10> rsfso7 <10-> rsfso10}{}
\DeclareMathAlphabet{\mycal}{OT1}{rsfso}{m}{n}

\newcommand{\?}{\;\!} 

\begin{document}
	
\thispagestyle{empty} 
	
\begin{abstract}
		We prove the existence of minimizers in the class of negative definite measures on compact subsets of momentum space in the homogeneous setting under several side conditions (constraints). The method is to employ Prohorov's theorem. Given a minimizing sequence of negative definite measures, we show that, under suitable side conditions, a unitarily equivalent subsequence thereof is bounded. 
		By restricting attention to compact subsets, from Prohorov's theorem  
		we deduce the existence of minimizers in the class of negative definite measures. 
\end{abstract}
	
\maketitle
	
\thispagestyle{empty} 
	
\tableofcontents
	
\section{Introduction}

In the physical theory of causal fermion systems, 
spacetime and the structures therein are described by a minimizer
of the so-called causal action principle
(for an introduction to the physical background and the mathematical context, we refer the interested reader to 
the textbook~\cite{cfs}, the survey articles~\cite{review, dice2014} as well as the web platform~\cite{link}). 
Given a causal fermion system~$(\H, \F, d\rho)$ together with a non-negative function~$\L : \F \times \F \to \R_0^+ := [0, \infty)$ (the \emph{Lagrangian}), the causal action principle is to minimize the \emph{action}~$\Sact$ defined as the double integral over the Lagrangian 
\[ \Sact (\rho) = \int_\F d\rho(x) \int_\F d\rho(y)\: \L(x,y) \]
under variations of the measure~$d\rho$ within the class of regular Borel measures on~$\F$ under suitable side conditions.
In order to work out the existence theory for minimizers,  
{\em{causal variational principles}} evolved as a mathematical generalization
of the causal action principle~\cite{continuum, jet}. The aim of the present paper is to give an alternative proof for the existence of minimizers for causal variational principles restricted to compact subsets in the \emph{homogeneous} setting. 

In order to put the present paper into the mathe\-matical context, in~\cite{pfp} it was proposed to
formulate physics by minimizing a new type of variational principle in spacetime.
The suggestion in~\cite[Section~3.5]{pfp} led to the causal action principle
in discrete spacetime, which was first analyzed mathematically in~\cite{discrete}.
A more general and systematic inquiry of causal variational principles on measure spaces was carried out
in~\cite{continuum}. In~\cite[Section~3]{continuum} the existence of minimizers for variational principles in indefinite inner product spaces is proven in the special case that the total spacetime volume as well as the number of particles therein are finite. Under the additional assumption that the kernel of the fermionic projector is \emph{homogeneous} in the sense that it only depends on the difference of two spacetime points, variational principles for {homogeneous} systems were considered in~\cite[Section 4]{continuum} in order to deal with an infinite number of particles in an infinite spacetime volume. More precisely, the main advantage in the homogeneous setting is that it allows for Fourier methods, thus giving rise to a natural correspondence between position and momentum space. As a consequence, one is led to minimize the causal action by varying in the class of negative definite measures, and the existence of minimizers on bounded subsets of momentum space is proven in~\cite[Theorem~4.2]{continuum}. The aim of this paper is to give an alternative proof of this existence result for compact subsets. In addition, the result is stated for additional side conditions (see Section~\ref{Section existence proof}) which were not considered in~\cite{continuum}.

The paper is organized as follows. In Section~\ref{Section Physical Background} we first outline some mathematical preliminaries (\S \ref{S Basic Definitions}) and afterwards recall causal variational principles in infinite spacetime volume (\S \ref{S Preliminaries}). In order to put the causal variational principles into the context of calculus of variations, in Section~\ref{Section CVP Homogeneous} we first introduce so-called operator-valued measures (\S \ref{S operator-valued measures}); afterwards, we consider variational principles on compact subsets of momentum space in the homogeneous setting (\S \ref{S variational principle}). In Section~\ref{Section existence proof}, we prove the existence of minimizers for the causal variational principle on compact subsets in the class of negative definite measures (Theorem~\ref{Theorem minimizer}). To this end we first show that, under appropriate side conditions, minimizing sequences of negative definite measures are bounded with respect to the total variation (\S \ref{S boundedness}). We then state a preparatory result which ensures the existence of weakly convergent subsequences (\S \ref{S preparatory}). This allows us to prove our main result (\S \ref{S existence}). Afterwards we show that the main result also holds in the case that a boundedness constraint is imposed (\S \ref{S boundedness constraint}). In this way, we give an alternative proof of~\cite[Theorem~4.2]{continuum} (Theorem~\ref{Theorem boundedness}). In the appendix we give a possible explanation for the side conditions under consideration (Appendix~\ref{Appendix justification}).

\section{Mathematical Preliminaries}\label{Section Physical Background}

\subsection{Mathematical Preliminaries and Notation}\label{S Basic Definitions}
To begin with, let us compile some fundamental definitions being of central relevance throughout this paper. For details we refer the interested reader to~\cite{bognar}, \cite{gohberg} and~\cite{langer}. Unless specified otherwise, we always let~$n \ge 1$ be a given integer. 

\begin{Def} 
	A mapping~$\prec . \mid . \succ \, \colon \C^n \times \C^n \to \C$ is called an \textbf{indefinite inner product} if the following conditions hold (cf.~\cite[Definition 2.1]{gohberg}):
	\begin{enumerate}[leftmargin=2em]
		\item[\rm{(i)}] $\prec y \mid \alpha x_1 + \beta x_2\succ \,= \alpha \prec y \mid x_1 \succ + \,\beta \prec y \mid x_2 \succ $ for all $x_1$, $x_2$, $y\in \C^n$, $\alpha$, $\beta \in \C$. 
		\item[\rm{(ii)}] $\prec x\mid y\succ \,= \,\overline{\prec y\mid x\succ}$ for all $x$, $y\in \C^n$. 
		\item[\rm{(iii)}] $\prec x\mid y\succ\,=0$ for all $y\in \C^n $ $\Longrightarrow$ $x=0$. 
	\end{enumerate} 
\end{Def}

\begin{Def}
	Let $V$ be a finite-dimensional complex vector space, endowed with an indefinite inner product~${\prec . \mid . \succ}$.
	Then $(V, \prec . \mid . \succ)$ is called an \textbf{indefinite inner product space}. 
\end{Def}

As usual, by~$\LL(V)$ we denote the set of (bounded) linear operators on a complex (finite-dimensional) vector space~$V$ of dimension~$n \in \N$. The adjoint of~$A \in \LL(V)$ with respect to the Euclidean inner product~$\langle \, . \, | \, . \, \rangle$ on~$V \simeq \C^n$ is denoted by~$A^{\dagger}$. On the other hand, whenever~$(V, \prec . \mid . \succ)$ is an indefinite inner product space, unitary matrices and the adjoint~$A^{\ast}$ (with respect to~$\prec . \mid . \succ$) are defined as follows. 

\begin{Def}
	Let $\prec . \mid . \succ$ be an indefinite inner product on~$V \simeq \C^n$, and let $S$ be the associated invertible hermitian matrix determined by~\cite[eq.\ (2.1.1)]{gohberg}, 
	\begin{align*}
	{\prec x \mid y \succ} = \langle S \? x \mid y \rangle \qquad \text{for all $x,y\in \C^n$} \:.
	\end{align*}
	Then for every~$A \in \LL(V)$, the adjoint of~$A$ (with respect to~$\prec . \mid . \succ$) is the unique matrix~$A^{\ast} \in \LL(V)$ which satisfies 
	\begin{align*}
	{\prec A \?x \mid y \succ} = {\prec x \mid A^{\ast} \? y \succ} \qquad \text{for all~$x,y \in V$} \:.
	\end{align*}
	A matrix~$A \in \LL(V)$ is called \textbf{self-adjoint (with respect to~$\prec . \mid . \succ$)} if and only if~$A = A^{\ast}$. 
	In a similar fashion, an operator~$U \in \LL(V)$ is said to be \textbf{unitary (with respect to~$\prec . \mid. \succ$)} if it is invertible and $U^{-1} = U^{\ast}$ (see~\cite[Section~4.1]{gohberg}). 
\end{Def}\noindent 
We remark that every non-negative matrix (with respect to~$\prec .\mid . \succ$) is self-adjoint (with respect to~$\prec . \mid . \succ$) and has a real spectrum (cf.~\cite[Theorem~5.7.2]{gohberg}). Moreover, the adjoint~$A^{\ast}$ of~$A \in \LL(V)$ satisfies the relation 
\begin{align*}
A^{\ast} = S^{-1} \: A^{\dagger} \: S 
\end{align*}
in view of~\cite[eq.~(4.1.3)]{gohberg} (where~$A^{\dagger}$ denotes the adjoint with respect to~$\langle \, . \, | \, . \, \rangle$ and~$A^{\ast}$ the adjoint with respect to~$\prec . \mid .\succ$). 
For details concerning self-adjoint operators (with respect to~$\prec . \mid . \succ$) we refer to~\cite{langer} and the textbook~\cite{bognar}.

In the remainder of this paper we will restrict attention exclusively to indefinite inner product spaces~$(V, \prec . \mid . \succ)$ with~$V \simeq \C^{2n}$ for some~$n \in \N$. 
It is convenient to work with a fixed pseudo-orthonormal basis $(\mathfrak{e}_i)_{i=1, \ldots, 2n}$ of~$V$ in which the inner product has the standard representation with a signature matrix $S$,
\begin{align}\label{(3.13)}
{\prec u \mid v \succ} = \langle u \mid Sv \rangle_{\C^{2n}} \qquad \text{with} \qquad S = \diag(\underbrace{1, \ldots, 1}_{\text{$n$ times}}, \underbrace{-1, \ldots, -1}_{\text{$n$ times}}) \:,
\end{align}
where $\langle \, . \, | \, . \, \rangle_{\C^{2n}}$ denotes the standard inner product on $\C^{2n}$. 
The signature matrix can be regarded as an operator on $V$, 
\begin{align}\label{(S)}
S = \begin{pmatrix}
\Id & 0 \\ 0 & -\Id
\end{pmatrix} \in \operatorname{Symm} V \:, 
\end{align}
where~$\operatorname{Symm} V$ denotes the set of symmetric matrices on~$V$ with respect to the spin scalar product
(also cf.~\cite[proof of Lemma 3.4]{continuum}). Without loss of generality we may assume that~$\mathfrak{e}_i = (0, \ldots, 0, 1, 0, \ldots, 0)^{\mathsf{T}}$ for all~$i= 1, \ldots, 2n$. 

In what follows, we denote Minkowski space by~$\scrM \simeq \R^4$ and momentum space by~$\hscrM \simeq \R^4$. Identifying~$\hscrM$ with Minkowski space~$\scrM$, the Minkowski inner product (of signature~$(+,-,-,-)$) can be considered as a mapping  
\begin{align*}
\langle . , . \rangle : \hscrM \times \scrM \to \R \:, \qquad (k,x) \mapsto \langle k, \xi \rangle = \eta_{\mu\nu} k^{\mu} \: \xi^{\mu} = k^0 \xi^0 - \sum_{i=1}^{3} k^{i} \: \xi^{i} 
\end{align*}
for all~${\xi = (\xi^0, \xi^1, \xi^2, \xi^3) \in \scrM}$ and~${k = (k^0, k^1, k^2, k^3) \in \hscrM}$ (with Minkowski metric~$\eta$, where we employed Einstein's summation convention, cf.~\cite[Chapter~1]{folland-qft}). 

In the remainder of this paper, let~$\hat{K} \subset \hscrM$ be a compact subset. 
By~$\B(\hat{K})$ we denote the Borel $\sigma$-algebra on~$\hat{K}$. The class of finite complex measures on~$\hat{K}$ is denoted by~$\mathbf{M}_{\C}(\hat{K})$. By~$C_c(\hscrM)$ we denote the set of continuous functions on~$\hscrM$ with compact support, whereas~$C_b(\hscrM)$ and~$C_0(\hscrM)$ indicate the sets of continuous functions on~$\hscrM$ which are bounded or vanishing at infinity, respectively. Since~$\hat{K}$ is compact, the sets~$C_c(\hat{K})$ and~$C_b(\hat{K})$ coincide.  
By~$L^1_{\textup{loc}}(\scrM)$ we denote the set of locally integrable functions on~$\scrM$ with respect to Lebesgue measure, denoted by~$d\mu$. Unless otherwise specified, we always refer to locally finite measures on the Borel $\sigma$-algebra as Borel measures in the sense of~\cite{gardner+pfeffer}. A Borel measure is said to be regular if it is inner and outer regular. Inner regular Borel measures are referred to as Radon measures~\cite{elstrodt}.

\subsection{Variational Principles in Infinite Spacetime Volume}\label{S Preliminaries} 
Before entering variational principles in infinite spacetime volume, 
let us briefly recall the concept of a Dirac sea as introduced by Paul Dirac in his paper~\cite{dirac1930theory}. In this article, he assumes that
\begin{quote}
	``(...) \emph{all the states of negative energy are occupied except perhaps a few of small velocity.} (...) \emph{Only the small departure from exact uniformity, brought about by some of the negative-energy states being unoccupied, can we hope to observe.} (...) We are therefore led to the assumption that \emph{the holes in the distribution of negative-energy electrons are the [positrons].}''
\end{quote}
Dirac made this picture precise in his paper~\cite{dirac1934discussion} by introducing a relativistic density matrix~$R(t, \vec{x}; t', \vec{x}')$ with $(t, \vec{x}), (t', \vec{x}') \in \R \times \R^3$ defined by
\begin{align*}
R(t, \vec{x}; t', \vec{x}') = \sum_{\text{$l$ occupied}} \Psi_l(t, \vec{x}) \: \overline{\Psi_l(t', \vec{x}')} \:.
\end{align*}
In analogy to Dirac's original idea, in~\cite{external} the kernel of the fermionic projector is introduced as the sum over all occupied wave functions 
\begin{align*}
P(x,y) = -  \sum_{\text{$l$ occupied}} \Psi_l(x) \: \overline{\Psi_l(y)} 
\end{align*}
for spacetime points~$x,y \in \scrM$ as outlined in~\cite{finster2011formulation}. A straightforward calculation shows that (see e.g.~\cite[\S 4.1]{finster+grotz}) the kernel of the fermionic projector takes the form
\begin{align}\label{(Dirac sea)}
P(x,y) = \int_{\hscrM} \frac{d^4k}{(2\pi)^4} \: (\slashed{k} + m) \: \delta(k^2-m^2) \: \Theta(-k^0) \: e^{-ik(x-y)} 
\end{align}
(where~$\delta$ denotes Dirac's delta distribution and~$\Theta$ is the Heaviside function).  
We refer to~$P(x,y)$ as the (unregularized) \emph{kernel of the fermionic projector of the vacuum} (cf.~\cite[eq.~(1.2.20) and eq.~(1.2.23)]{cfs} as well as~\cite[eq.~(4.1.1)]{pfp}; this object already appears in~\cite{finster1996ableitung}). We also refer to~\eqref{(Dirac sea)} as a \emph{completely filled Dirac sea}.    
The kernel of the fermionic projector \eqref{(Dirac sea)} is the starting point for the analysis in \cite[Section 4]{continuum}. In order to deal with systems containing an infinite number of particles in an infinite spacetime volume, the main simplification in~\cite{continuum} is to assume that the kernel of the fermionic projector~\eqref{(Dirac sea)} is \emph{homogeneous} in the sense that~$P(x,y)$ only depends on the difference vector~$y-x$ for all spacetime points~$x,y \in \scrM$. The underlying homogeneity assumption~$P(x,y) = P(y-x)$ for all~$x,y \in \scrM$ is referred to as ``homogeneous regularization of the vacuum'' (cf.~\cite[eq.~(4.1.2)]{pfp} and the explanations thereafter; also see~\cite[Assumption~3.3.1]{cfs}). 
Introducing~$\xi = \xi(x,y) := y-x$ for all~$x,y \in \scrM$ and
\begin{align*}
\hat{P}(k) = (\slashed{k} + m) \: \delta(k^2 - m^2) \: \Theta(-k^0) 
\end{align*}
for all~$k \in \hscrM$,  
the fermionic projector~\eqref{(Dirac sea)} can be written as a Fourier transform, 
\begin{align*}
P(x,y) = \int_{\hscrM} \frac{d^4k}{(2\pi)^4} \: \hat{P}(k) \: e^{i\langle k, \xi \rangle}
\end{align*}
(for details concerning the Fourier transform we refer to~\cite{folland-qft}). 
In order to arrive at a measure-theoretic framework, it is convenient to regard~$\hat{P}(k) \: d^4k/(2\pi)^4$ 
as a Borel measure~$d\nu$ on~$\hscrM$, taking values in $\LL(V)$. In particular, the measure
\begin{align}\label{(4.1)}
d\nu(k) = (\slashed{k} + m) \: \delta(k^2 - m^2) \: \Theta(-k^0) \: d^4k
\end{align}
has the remarkable property that~$- d\nu$ is positive in the sense that
\begin{align}\label{(4.2')}
{\prec u \mid - \nu(\Omega) \: u \succ} \ge 0 \qquad \text{for all $u \in \C^4$}
\end{align}
with respect to the ``spin scalar product''~$\prec . \mid . \succ$ on~$\C^4$ introduced in~\S \ref{S Basic Definitions}.\footnote{
	In order to see this, we make use of the fact that the Dirac matrices anti-commute, i.e.
	\begin{align*}
	\gamma^{\mu} \gamma^{\nu} = - \gamma^{\nu} \gamma^{\mu} \qquad \text{whenever~$\mu \not= \nu$} \:. 
	\end{align*}
	Thus for every~$k \in \hscrM$ with~$k = (k^0, \vec{k})$, the operators~$p_{\pm}(\vec{k})$ given by~\cite[eq.~(2.13)]{oppiocfs} satisfy
	\begin{align*}
	p_{\pm}(\vec{k}) \: \gamma^0 = \frac{\slashed{k}+m}{2k^{0}}\: \gamma^0|_{k^0 = \pm \omega(\vec{k})} = \gamma^0 \: p_{\pm}(-\vec{k})
	\end{align*}
	with~$\omega(\vec{k}) = \sqrt{\vec{k}^2 + m^2}$. 
	Applying the fact that~$p_{\pm}(\vec{k})$ is idempotent and symmetric with respect to the Euclidean scalar product~$\langle \,.\, | \,  .\, \rangle_{\C^4}$ on~$\C^{4}$ (cf.~\cite[Proposition~2.14]{oppiocfs}), the calculation
	\begin{align*}
	{\prec u \mid (\slashed{k}+m) \: u \succ} 
	= 2 k^0 {\prec u \mid \gamma^0 p_{\pm}(-\vec{k}) \: u \succ} = 2k^0 \langle u \mid p_{\pm}(-\vec{k})^2 \: u \rangle_{\C^4} =2k^0 \langle p_{\pm}(-\vec{k}) \: u \mid p_{\pm}(-\vec{k}) \: u \rangle_{\C^4} 
	\end{align*}
	for any~$u \in \C^4$ implies that
	\begin{align*}
	\prec \cdot \mid (\slashed{k} + m) \: \cdot \succ \quad \text{is} \quad \left\{ \begin{array}{cl}
	\text{positive semidefinite} \qquad &\text{if $k^0 > 0$} \\ [0.2em]
	\text{negative semidefinite} \qquad &\text{if $k^0 < 0$} \:. 
	\end{array} \right.
	\end{align*} 
	Introducing the set~$\Omega^- = \Omega \cap \{k^0 < 0 : k = (k^0, \vec{k}) \in \hscrM \}$ for any~$\Omega \in \B(\hscrM)$, for all~$u \in V$ we obtain
	\begin{align*}
	{\prec u \mid - \nu(\Omega) \: u \succ} &= {\prec u \mid - \int_{\Omega} (k_j \: \gamma^j + m) \: \delta\left(\langle k , k \rangle - m^2 \right) \: \Theta(-k^0) \: d^4k \: u \succ} \\
	&= \int_{\Omega^-} \underbrace{-\prec u \mid (\slashed{k} + m) u \succ}_{\text{$\ge 0$}} \: \delta\left(\langle k , k \rangle - m^2\right) \: d^4k \ge 0 \:.
	\end{align*}
	Therefore, positivity \eqref{(4.2')} is a consequence of the corresponding behavior of the operator~$(\slashed{k} + m)$.}

In order to avoid ultraviolet problems, caused by measures of the form~\eqref{(4.1)}, one is led to restrict attention to compact subsets of momentum space~\cite{continuum}. Moreover, generalizing~$(\C^4, \prec. \mid . \succ)$ to some indefinite inner product space~$(V, \prec. \mid . \succ)$ of 
dimension~$2n$, the above observations motivate the following definition (see~\cite[Definition~4.1]{continuum}).     

\begin{Def}\label{Definition 4.1}
	A vector-valued Borel measure~$d\nu$ on a compact set~$\hat{K} \subset \hscrM$ taking values in~$\LL(V)$ is called a \textbf{negative definite measure} on~$\hat{K}$ with values in~$\LL(V)$ whenever~$d\prec u \mid -\nu \, u \succ$ is a positive finite measure for all~$u \in V$.  
	By~$\ndm$ we denote the class of negative definite measures on~$\hat{K}$ taking values in~$\LL(V)$.
\end{Def}

In terms of a negative definite measure~$d\nu$, the \emph{kernel of the fermionic projector} is then introduced by
\begin{align*}
P(\xi) := \int_{\hat{K}} e^{i\langle p, \xi \rangle} \: d\nu(p) \qquad \text{for all~$\xi \in \scrM$} \:. 
\end{align*}
In order to clarify the dependence on~$d\nu$, we also write~$P[\nu]$.  
For every~$\xi \in \scrM$, the \emph{closed chain} is defined by~$A(\xi) := P(\xi) \? P(-\xi)$. In order to emphasize that the closed chain depends on~$d\nu$, we also write~$A[\nu]$. According to~\cite[eq.~(3.7)]{continuum}, the spectral weight~$|A|$ of an operator~$A \in \LL(V)$ is given by the sum of the absolute values of the eigenvalues of~$A$, 
\begin{align*}
|A| = \sum_{i=1}^{2n} |\lambda_i| \:,
\end{align*}
where by~$\lambda_i$ we denote the eigenvalues of~$A$, counted with algebraic multiplicities. In analogy to~\cite[eq.~(3.8)]{continuum}, for every~$\xi \in \scrM$ the Lagrangian is introduced via
\begin{align*}
\L[A(\xi)] := |A(\xi)^2| - \frac{1}{2n} |A(\xi)|^2 \:. 
\end{align*}
Defining the \emph{action}~$\Sact$ according to~\cite[eq.~(4.5)]{continuum} by
\begin{align*}
\Sact : \ndm \to [0, + \infty] \;, \qquad \Sact(\nu) := \int_{\scrM} \L[A(\xi)] \: d\mu(\xi) \:,
\end{align*}
the \emph{causal variational principle in the homogeneous setting} is to  
\begin{align*}
\boxed{\qquad \phantom{\int}\text{minimize $\Sact(\nu)$ by suitably varying $d\nu$ in $\ndm$} \:. \qquad \phantom{\int}}
\end{align*}
Introducing the functional~$\T$ by
\begin{align*}
\T : \ndm \to [0, + \infty] \;, \qquad \T(\nu) := \int_{\scrM} |A(\xi)|^2 \: d\mu(\xi) \:, 
\end{align*}
the main result in~\cite[Section~4]{continuum} can be stated as follows (see~\cite[Theorem~4.2]{continuum}): 

\begin{Thm}\label{Theorem 4.2}
	Let~$(d\nu_k)_{k \in \N}$ be a sequence of negative definite measures on the bounded set~$\hat{K} \subset \hscrM$ such that the functional~$\T$ is bounded by some constant~$C > 0$, i.e.
	\begin{align*}
	\T(\nu_k) \le C \qquad \text{for all~$k \in \N$} \:.
	\end{align*}
	Then there is a subsequence~$(d\nu_{k_{\ell}})_{\ell \in \N}$ as well as a sequence of unitary transformations~$(U_{\ell})_{\ell \in \N}$ on~$V$ (with respect to~$\prec . \mid . \succ$) such that the measures~$U_{\ell} \: d\nu_{k_{\ell}} \: U_{\ell}^{-1}$ converge weakly to a negative definite measure~$d\nu$ with the properties
	\begin{align*}
	\T(\nu) \le \liminf_{k \to \infty} \T(\nu_k) \:, \qquad \Sact(\nu) \le \liminf_{k \to \infty} \Sact(\nu_k) \:. 
	\end{align*}
\end{Thm}

Theorem~\ref{Theorem 4.2} is stated as a compactness result. Applying it to a minimizing sequence yields statements similar to~\cite[Theorem~2.2 and Theorem~2.3]{continuum}, asserting that the functional~$\Sact$ attains its minimum.

Given a negative definite measure~$d\nu$, the complex measure~${d\prec u \mid \nu \: v \succ} \in \mathbf{M}_{\C}(\hat{K})$ is defined by polarization for all $u, v \in V$, 
\begin{align}\label{(polarization)}
\begin{split}
{d\prec u \mid \nu \: v \succ} &:= \frac{1}{4} \big\{ {d\prec u + v \mid \nu \: (u+v) \succ} + i \: {d\prec u + iv \mid \nu \: (u+iv) \succ}  \\
&\qquad -\:  {d\prec u - v \mid \nu \: (u-v) \succ} - i \: {d\prec u - iv \mid \nu \: (u-iv) \succ} \big\} 
\end{split}
\end{align} 
(see e.g.~\cite[eq.~(2.2.6)]{gohberg}, also cf.~\cite[Section~VIII.3]{reed+simon}). 
Following~\cite[Definition~A.16]{langerma}, 
we define integration with respect to negative definite measures as follows: 
\begin{Def}\label{Definition integral negative}
	Let~$(V, \prec . \mid . \succ)$ be an indefinite inner product space and let~$d\nu$ be a negative definite measure. Moreover, let~$f : \hat{K} \to \C$ be a bounded Borel measurable function. For all~$u, v \in V$, integration with respect to~$d\nu$ is defined by 
	\begin{align*}
	{\prec u \mid \left(\int_{\hat{K}} f(k) \: d\nu(k) \right) \: v \succ} := \int_{\hat{K}} f(k) \: d\prec u \mid \nu(k) \: v \succ \:. 
	\end{align*}
\end{Def}
A similar definition in terms of operator-valued measures is stated below (see Definition~\ref{Definition integral operator}). For a connection to spectral theory we refer to~\cite[Chapter~31]{lax}.

\section{Causal Variational Principles in the Homogeneous Setting}\label{Section CVP Homogeneous}
\subsection{Operator-Valued Measures}\label{S operator-valued measures}
In order to deal with causal variational principles in the homogeneous setting in sufficient generality, 
this subsection is devoted to put the definition of negative definite measures (see Definition~\ref{Definition 4.1}) into the context of calculus of variations. 
More precisely, as explained in~\S \ref{S Preliminaries}, the variational principle as introduced 
in~\cite[Section~4]{continuum} is to minimize the causal action~$\Sact$ in the class of negative definite measures. 
Unfortunately, in view of~\eqref{(4.2')}, the set of negative definite measures does not form a vector space, whereas in calculus of variations one usually considers functionals on a real, locally convex vector space (for details we refer to~\cite[Section~43.2]{zeidlerIII}). Hence in order to obtain a suitable framework, we first introduce operator-valued measures, which can be regarded as a generalization of negative definite measures, thus providing the basic structures required for the calculus of variations (see Lemma~\ref{Lemma OVM Banach space} below). 
Concerning the connection to vector-valued measures we refer to~\cite{diestel+uhl}. 

Then operator-valued measures on a compact subset~$\hat{K} \subset \hscrM$ with values in~$\LL(V)$ are introduced as a generalization of negative definite measures 
(see Definition~\ref{Definition 4.1}) in the following way: 

\begin{Def}\label{Definition Operator-valued measure}
	A vector-valued measure~$d\omega$ on~$\B(\hat{K})$ taking values in~$\LL(V)$ is called an \textbf{operator-valued measure} on~$\hat{K}$ with values in~$\LL(V)$ whenever~$d\prec u \mid \omega \: v \succ$ is a finite complex measure in~$\mathbf{M}_{\C}(\hat{K})$ for all~$u,v \in V$. 
\end{Def}

Whenever~$\hat{K}$ and $V$ are understood, the class of operator-valued measures on~$\hat{K}$ with values in~$\textup{L}(V)$ shall be denoted by~$\ovm$. 

In what follows, the variation of an operator-valued measure plays a central role: 

\begin{Def}\label{Definition variation}
	Given an operator-valued measure~$d\omega \in \ovm$, the \textbf{variation} of~$d\omega$, denoted by~$d|\omega|$, is defined by
	\begin{align*}
	d|\omega| := \sum_{i,j=1}^{2n} d\left|\prec \mathfrak{e}_i \mid \omega \: \mathfrak{e}_j \succ \right| \:,
	\end{align*}
	where~$d\left| \, . \, \right|$ denotes the variation of a complex measure. 
	Moreover, the \textbf{total variation} of~$d\omega$, denoted by~$d\|\omega\|$, is given by
	\begin{align}\label{(norm of a negative definite measure)}
	d\|\omega\| := d|\omega|(\hat{K}) = \sum_{i,j=1}^{2n} d\left|\prec \mathfrak{e}_i \mid \omega \: \mathfrak{e}_j \succ \right| (\hat{K}) \:. 
	\end{align}
\end{Def}\noindent
We point out that the variation as given by Definition~\ref{Definition variation} crucially depends on the pseudo-orthogonal~$(\mathfrak{e}_i)_{i=1, \ldots, 2n}$ basis of~$V$. Nevertheless, the set of operator-valued measures~$\ovm$ is a Banach space with respect to the total variation: 

\begin{Lemma}\label{Lemma OVM Banach space}
	The total variation~$d\|\cdot\|$ given by~\eqref{(norm of a negative definite measure)} 
	defines a norm on $\ovm$ in such a way that~$(\ovm, d\|\cdot\|)$ is a complex Banach space. 
	In particular,~$(\ovm, d\|\cdot\|)$ is a real, locally convex vector space. 
\end{Lemma}
\begin{proof}
	For the first part of the statement see the proof of~\cite[Corollary~5.3]{langerma}. 
	In order to show that~$\ovm$ is a Banach space, 
	let us consider a Cauchy sequence of operator-valued measures~$(d\omega_k)_{k \in \N}$ with respect to the norm~\eqref{(norm of a negative definite measure)}, that is, $d\|\omega_k - \omega_m\| \to 0$ as~$k,m \to \infty$. Our task is to prove that its limit, denoted by~$d\omega$, exists and that~$d\omega$ is contained in~$\ovm$. Assuming that~$(\mathfrak{e}_i)_{i=1, \ldots, 2n}$ is a pseudo-orthonormal basis of~$V$ satisfying~\eqref{(3.13)}, 
	from~\eqref{(norm of a negative definite measure)} we deduce that 
	\begin{align*}
	\lim_{k,m \to \infty} d\| {\prec \mathfrak{e}_i \mid \left(\omega_k - \omega_m \right) \mathfrak{e}_j \succ} \| = 0 \qquad \text{for all $i, j = 1, \ldots, 2n$} \:.
	\end{align*}
	Consequently, each sequence~$({d \prec \mathfrak{e}_i \mid \omega_k \: \mathfrak{e}_j \succ})_{k \in \N}$ is a Cauchy sequence of complex measures in~$\mathbf{M}_{\C}(\hat{K})$ for all $i,j \in \{1, \ldots, 2n \}$. Since~$\mathbf{M}_{\C}(\hat{K})$ is a complex Banach space with respect to the total variation~$d\|\cdot\|$ in virtue of~\cite[Aufgabe~VII.1.7]{elstrodt}, there is a complex measure~$d\omega_{i,j} \in \mathbf{M}_{\C}(\hat{K})$, being the unique limit of~$({d \prec \mathfrak{e}_i \mid \omega_k \: \mathfrak{e}_j \succ})_{k \in \N}$ for all $i,j \in \{1, \ldots, 2n \}$.
	
	Next, for all $i,j \in \{1, \ldots, 2n \}$, the complex measures $d\omega_{i,j}$ in~$\mathbf{M}_{\C}(\hat{K})$ give rise to an operator-valued measure $d\omega$ on~$\hat{K}$ with values in~$\LL(V)$ in such a way that, for all~$i,j = 1, \ldots, 2n$, we are given~${d \prec \mathfrak{e}_i \mid \omega \: \mathfrak{e}_j \succ} = d\omega_{i,j}$. 
	More precisely, defining the operator~$\omega(\Omega) \in \LL(V)$ for any~$\Omega \in \B(\hat{K})$ by
	\begin{align*}
	\omega(\Omega) := \begin{pmatrix}
	&\omega_{1,1}(\Omega) &\cdots &\omega_{1, 2n}(\Omega) \\
	&\vdots & \ddots&\vdots \\
	&\omega_{n,1}(\Omega) &\cdots &\omega_{n,2n}(\Omega) \\
	&-\omega_{n+1,1}(\Omega) &\cdots &-\omega_{n+1,2n}(\Omega) \\
	&\vdots & \ddots&\vdots\\
	&-\omega_{1,2n}(\Omega) & \cdots &-\omega_{2n,2n}(\Omega)
	\end{pmatrix} \in \LL(V) \:,
	\end{align*}
	we obtain a mapping $d\omega \colon \B(\hat{K}) \to \LL(V)$ such that ${d\prec \mathfrak{e}_i \mid \omega \: \mathfrak{e}_j \succ} = d\omega_{i,j} \in \mathbf{M}_{\C}(\hat{K})$ for all $i,j \in \{1, \ldots, 2n \}$. Since $(\mathfrak{e}_i)_{i=1, \ldots, 2n}$ is a basis of $V$, for any $\Omega \in \B(\hat{K})$ and arbitrary elements $u = \sum_{i=1}^{2n} \alpha_i \: \mathfrak{e}_i$, $v = \sum_{j=1}^{2n} \beta_j \: \mathfrak{e}_j \in V$ we arrive at
	\begin{align*}
	{\prec u \mid \omega(\Omega) \: v \succ} = \sum_{i,j=1}^{2n} \overline{\alpha}_i \: \beta_j \: {\prec \mathfrak{e}_i \mid \omega(\Omega) \: \mathfrak{e}_j \succ} = \sum_{i,j = 1}^{2n} \overline{\alpha_i}\: \beta_j \: \omega_{i,j}(\Omega) \:.
	\end{align*}
	The fact that~$\mathbf{M}_{\C}(\hat{K})$ is a complex Banach space implies that
	\begin{align*}
	{d \prec u \mid \omega \: v \succ} = \sum_{i,j = 1}^{2n} \overline{\alpha_i} \: \beta_j \: d\omega_{i,j} \in \mathbf{M}_{\C}(\hat{K}) \qquad \text{for all $u,w \in V$} \:. 
	\end{align*}
	This shows that~$d\omega \in \ovm$ is an operator-valued measure in view of Definition~\ref{Definition Operator-valued measure}. Thus~$(\ovm, d\|\cdot\|)$ is a complex Banach space with respect to the norm~$d\|\cdot\|$ defined by~\eqref{(norm of a negative definite measure)}. Since each norm induces a corresponding Fréchet metric,~$(\ovm, d\|\cdot\|)$ can be regarded as a metric space. In particular, each complex vector space is a real one, and each Banach space is locally convex. This completes the proof. 
\end{proof}

\begin{Remark}
	The set of negative definite measures~$\ndm$ clearly is a subset of the vector space $\ovm$. 
	However, $\ndm$ itself is \emph{not} a vector space (see~\cite[Remark 5.6]{langerma}), but a \emph{cone}, i.e.~a closed subset under multiplication with positive real numbers. 
\end{Remark}

Next, let us introduce the support of operator-valued measures as follows: 

\begin{Def} 
	We define the \textbf{support} of an operator-valued measure~$d\omega$ in~$\ovm$ as the support of its variation measure~$d|\omega|$, i.e.
	\begin{align*}
	\supp d\omega := \supp d|\omega| = \hat{K} \setminus \bigcup \left\{U \subset \hat{K} : \text{$U$ open and~$d|\omega|(U) = 0$} \right\} \:. 
	\end{align*}
\end{Def}\noindent
Since~$d|\omega|$ is a locally finite measure on a locally compact Polish space, 
we conclude that~$d|\omega|$ is regular and has support,~$d|\omega|(\hat{K} \setminus \supp d|\omega|) = 0$. 

In a similar fashion, following~\cite[Definition~7.1.5]{bogachev}, an operator-valued measure~$d\omega$ is called \emph{regular} if and only if~$d|\omega|$ is regular. Moreover, the measure~$d\omega$ is said to be \emph{tight} if for every~$\varepsilon > 0$ there is a compact set~$K_{\varepsilon} \subset \hat{K}$ such that~$d|\omega|(\hat{K} \setminus K_{\varepsilon}) < \varepsilon$ (cf.~\cite[Definition~7.1.4]{bogachev}). Clearly, whenever~$\hat{K} \subset \hscrM$ is compact, every operator-valued measure on~$\hat{K}$ is tight. 

\begin{Def}\label{Definition integral operator}
	In analogy to negative definite measures (see Definition~\ref{Definition integral negative}), for any bounded Borel measurable function~$f : \hat{K} \to \C$ we define integration with respect to operator-valued measures~$d\omega$ by
	\begin{align*}
	{\prec u \mid \left(\int_{\hat{K}} f(k) \: d\omega(k) \right) \: v \succ} := \int_{\hat{K}} f(k) \: d\prec u \mid \omega(k) \: v \succ \qquad \text{for all~$u, v \in V$} \:. 
	\end{align*}
\end{Def}

Let us finally state the definition of weak convergence of operator-valued measures, which will be required later on (see \S \ref{S existence} below). 

\begin{Def}\label{Definition weak convergence}
	We shall say that a sequence of operator-valued measures~$(d\omega_k)_{k \in \N}$ in~$\ovm$ \textbf{converges weakly} to some operator-valued measure~$d\omega$ if and only if
	\begin{align*}
	\lim_{k \to \infty} \int_{\hat{K}} f \: d\: {\prec u \mid \omega_k \: v \succ} = \int_{\hat{K}} f \: d\prec u \mid \omega \: v \succ \qquad \text{for all~$u, v \in V$ and~$f \in C_b(\hat{K})$} \:. 
	\end{align*}
	We write symbolically~$d\omega_k \rightharpoonup d\omega$. 
\end{Def}

Whenever~$d\nu \in \ndm$ is a negative definite measure, we recall that, for all~$u,v \in V$, the complex measure~${d\prec u \mid \nu \: v \succ}$ in~$\mathbf{M}_{\C}(\hat{K})$ is defined by polarization~\eqref{(polarization)}. Thus a sequence of negative definite measures~$(d\nu)_{k \in \N}$ converges weakly to some negative definite measure~$d\nu \in \ndm$ if and only if 
\begin{align*}
\lim_{k \to \infty} \int_{\hat{K}} f \: d\: {\prec u \mid \nu_k \: u \succ} = \int_{\hat{K}} f \: d\prec u \mid \nu \: u \succ \qquad \text{for all~$u \in V$ and~$f \in C_b(\hat{K})$} \:. 
\end{align*}
By polarization~\eqref{(polarization)} we then conclude that
\begin{align*}
\lim_{k \to \infty} \int_{\hat{K}} f \: d\: {\prec u \mid \nu_k \: v \succ} = \int_{\hat{K}} f \: d\prec u \mid \nu \: v \succ \qquad \text{for all~$u, v \in V$ and~$f \in C_b(\hat{K})$} 
\end{align*}
in accordance with Definition~\ref{Definition weak convergence}. 

Note that, with the very same reasoning, the definitions and results stated in this section can be generalized to operator-valued measures on whole momentum space.

\subsection{Causal Variational Principles on Compact Subsets}\label{S variational principle}
After these technical preliminaries, let us now return to causal variational principles in the homogeneous setting.  
Motivated by~\eqref{(Dirac sea)}, the fermionic projector~$P(x,y)$ in the homogeneous setting takes the form
\begin{align*}
P(x,y) = \int_{\hat{K}} e^{ik(y-x)} \: d\nu(k) 
\end{align*}
for all~$x, y \in \scrM$, 
where the measure~$d\nu$ is given by~\eqref{(4.1)}. 
Generalizing~$d\nu$ according to~\S \ref{S Preliminaries} and~$\S \ref{S operator-valued measures}$ to operator-valued measures, 
for a given operator-valued measure~$d\omega$ on $\hat{K}$ with values in $\textup{L}(V)$ and all~$x,y \in \scrM$ we introduce the \emph{kernel of the fermionic projector} by 
\begin{align*}
P(x,y) : V \to V, \qquad P(x,y):= \int_{\hat{K}} {e}^{ik(y-x)} \: {d}\omega(k) \:.
\end{align*}
In order to emphasize the dependence on the operator-valued measure $d\omega$, we also write~$P[\omega](x,y)$. 
As~$P(x,y)$ is supposed to be \emph{homogeneous}, 
only the difference of two spacetime points $x,y \in \scrM$ matters; denoting the difference vector by~$\xi = y-x \in \scrM$, the kernel of the fermionic projector reads 
\begin{align}\label{(fermionic projector)}
P(\xi) : V \to V, \qquad P(\xi) = \int_{\hat{K}} {e}^{ik\xi}\: {d}\omega(k) \:.
\end{align} 
The first step in order to 
set up the variational principle is
to form the \emph{closed chain}, which (as motivated by \cite[\S 3.5]{pfp}) for any~$\xi \in \scrM$ is defined as the mapping
\begin{align*}
A(\xi): V \to V, \qquad A(\xi):= P(\xi)\,P(-\xi) \:.
\end{align*}
We also write $A[\omega](\xi)$ in order to clarify the dependence of the closed chain on the operator-valued measure $d\omega$. Next, given a linear operator~$A : V \to V$, we define the \emph{spectral weight} by 
\begin{align*}
\left|A\right|:= \sum_{i=1}^{2n} \left|\lambda_i\right| \:, 
\end{align*}
where by~$(\lambda_i)_{i=1,\ldots,2n}$ we denote the eigenvalues of the operator~$A$, counted with algebraic multiplicities. 
In this way, the spectral weight furnishes a connection between endomorphisms and scalar functionals. 

In order to set up a real-valued variational principle on the set of operator-valued measures, for every~$d\omega \in \ovm$ we introduce the \emph{Lagrangian}
\begin{align*}
\L[\omega] : \scrM \to \R_0^+ \:, \qquad \L [{\omega}](\xi) := \left|A[{\omega}](\xi)^2\right|-\frac{1}{2n} \left|A[{\omega}](\xi)\right|^2 \:.
\end{align*}
Defining the \emph{causal action}~$\Sact :\ovm \to \R_0^+ \cup \{+ \infty \}$ by
\begin{align*}
\Sact(\omega) := \int_{\scrM} \L [{\omega}](\xi) \: {d}\mu (\xi) \:,
\end{align*}
the \emph{causal variational principle in the homogeneous setting} is to  
\begin{align}\label{(causal action principle homogeneous)}
\boxed{\qquad \phantom{\int}\text{minimize $\Sact(\nu)$ by suitably varying $d\nu$ in $\ndm$} \:. \qquad \phantom{\int}}
\end{align}

In order to exclude trivial minimizers, we impose the \emph{trace constraint}
\begin{align}\label{(c)}
\Tr_V \big(\nu(\hat{K}) \big) = c
\end{align}
for some~$c > 0$. Additionally, 
for~$f > 0$ 
we shall either introduce the constraint 
\begin{align}\label{(f)}
\Tr_V \big(-S\nu(\hat{K}) \big) \le f
\end{align}
(where~$S$ denotes the signature matrix~\eqref{(S)}) 
or the side condition
\begin{align}\label{(swf)}
|\nu(\hat{K})| \le f \:.
\end{align}
A motivation for the constraints~\eqref{(c)}--\eqref{(swf)} can be found in Appendix~\ref{Appendix justification}. For the connection to the boundedness constraint as considered in~\cite[Section~4]{continuum} we refer to~\S \ref{S boundedness constraint} below. 

\begin{Def}\label{Def cvp} 
	Given a subset~$N \subset \ndm$, 
	the \textbf{causal variational principle in the homogeneous setting} is to  
	\begin{align}\label{(cvp homogeneous)}
	\text{minimize $\Sact(\nu)$ by varying $d\nu$ in $N \subset \ndm$} \:. 
	\end{align}  
\end{Def}
Concerning the side conditions~\eqref{(c)}--\eqref{(swf)}, the subset~$N$ takes either the form
\begin{align*}
N &= \left\{ d\nu \in \ndm : \text{$d\nu$ satisfies conditions~\eqref{(c)} and~\eqref{(f)}} \right\} \qquad \text{or} \\ 
N &= \left\{ d\nu \in \ndm : \text{$d\nu$ satisfies conditions~\eqref{(c)} and~\eqref{(swf)}} \right\} \:. 
\end{align*}

In agreement with~\cite[Definition~43.4]{zeidlerIII}, we define a minimizer for~$\Sact$ as follows: 
\begin{Def}\label{Definition minimizer}
	A negative definite measure~$d\nu \in N$ is said to be a \textbf{minimizer} for the causal variational principle~\eqref{(cvp homogeneous)} if and only if  
	\begin{align*}
	\Sact(\tilde{\nu}) \ge \Sact(\nu) \qquad \text{for all~$d\tilde{\nu} \in N$} \:. 
	\end{align*}
\end{Def}
For further details concerning the calculus of variations we refer to~\cite[Chapter~43]{zeidlerIII}.

\section{Existence of Minimizers on Compact Subsets}\label{Section existence proof}
This section is devoted to developing the existence theory for minimizers of the causal action principle~\eqref{(causal action principle homogeneous)} for given $c, f > 0$ either with respect to the constraints
\begin{align}\label{(constraints)}
\Tr_V \big(\nu(\hat{K}) \big) = c \qquad \text{and} \qquad \Tr_V\big(-S\nu(\hat{K}) \big) \le f 
\end{align}
or with respect to the side conditions 
\begin{align}\label{(constraints spectral)}
\Tr_V \big(\nu(\hat{K}) \big) = c \qquad \text{and} \qquad |\nu(\hat{K})| \le f \:.
\end{align}

The main result of this section can be stated  as follows:

\begin{Thm}\label{Theorem minimizer}
	Let $(d\nu^{(j)})_{j \in \N}$ be a minimizing sequence of negative definite measures in~$\ndm$ of the causal variational principle~\eqref{(causal action principle homogeneous)} with respect to the constraints~\eqref{(constraints)} or~\eqref{(constraints spectral)}, respectively. Then there exists a sequence of unitary operators~$(U_j)_{j \in \N}$ on $V$ (with respect to $\prec . \mid . \succ$) and a subsequence~$(d\nu^{(j_k)})_{k \in \N}$ such that~$(U_{j_k} \: d\nu^{(j_k)} \: U_{j_k}^{-1})_{k \in \N}$ converges weakly to some non-trivial negative definite measure~$d\nu \not= 0$. Moreover, 
	\begin{align*}
	\Sact(\nu) \le \liminf_{k \to \infty} \Sact(\nu^{(j_k)}) \:,
	\end{align*}
	and the limit measure~$d\nu \in \ndm$ satisfies the side conditions
	\begin{align}\label{(limiting conditions)}
	\Tr_V \big(\nu(\hat{K}) \big) = c \qquad \text{and} \qquad \Tr_V\big(-S\nu(\hat{K}) \big) \le f 
	\end{align}
	or 
	\begin{align}\label{(limiting conditions spectral)}
	\Tr_V \big(\nu(\hat{K}) \big) = c \qquad \text{and} \qquad |\nu(\hat{K})| \le f \:, 
	\end{align}
	respectively (with positive constants~$c,f > 0$).  
	In particular, the limit measure~$d\nu$ is a non-trivial minimizer of the causal variational principle~\eqref{(causal action principle homogeneous)} with respect to the side conditions~\eqref{(constraints)} or~\eqref{(constraints spectral)}, respectively. A fortiori, the above statements remain true in case that ``$\le$'' in~\eqref{(constraints)}, \eqref{(constraints spectral)} and~\eqref{(limiting conditions)}, \eqref{(limiting conditions spectral)} is replaced by ``$=$''. 
\end{Thm}

The remainder of this section is devoted to the proof of Theorem~\ref{Theorem minimizer}. The key idea for proving Theorem~\ref{Theorem minimizer} is essentially to apply Prohorov's theorem (see e.g.~\cite[Section~8.6]{bogachev}). 
To this end, we proceed in several steps.  
Given a minimizing sequence of negative definite measures which satisfies the side conditions~\eqref{(constraints)} or~\eqref{(constraints spectral)}, 
we first prove boundedness of a unitarily equivalent subsequence thereof (\S \ref{S boundedness}).  
The proof of Theorem~\ref{Theorem minimizer} is completed afterwards (\S \ref{S existence}).  
Once this is accomplished, we show that Theorem~\ref{Theorem minimizer} also applies in the case that a boundedness constraint is imposed (\S \ref{S boundedness constraint}). 

\subsection{Boundedness of Minimizing Sequences}\label{S boundedness}
Let us assume that~$(d\nu^{(k)})_{k \in \N}$ is a sequence of negative definite measures in~$\ndm$, 
either satisfying 
\begin{align*}
\Tr_V(-S\nu^{(k)}(\hat{K})) \le f \qquad \text{or} \qquad |\nu^{(k)}(\hat{K})| \le f
\end{align*}
for all~$k \in \N$ and some positive constant~$f > 0$ (and~$|\cdot|$ denotes the spectral weight). The aim of this subsection is to show that in both cases, there exists a sequence of unitary matrices~$(U_k)_{k \in \N}$ in~$\LL(V)$ (with respect to~$\prec . \mid . \succ$) such that the resulting sequence~${(U_k \? d\nu^{(k)} \? U_k^{-1})_{k \in \N}}$ is bounded in~$\ndm$ (with respect to the norm~\eqref{(norm of a negative definite measure)}). In particular, whenever the first condition is imposed, it eventually turns out that one can choose~$U_k = \Id_V$ for all~$k \in \N$. 
In preparation, let us state the following results:

\begin{Prp}\label{Proposition spectrum}
	For all~$B,C \in \LL(V)$, the operator products~$BC$ and~$CB$ have the same spectrum.
\end{Prp}
\begin{proof}
	Follow the arguments in~\cite[Section~3]{discrete} or cf.~\cite[eq.~(3.5.6)]{pfp}.
\end{proof}

\begin{Lemma}\label{Lemma spectrum}
	Assume that~$U \in \LL(V)$ is unitary (with respect to~$\prec . \mid . \succ$), and let~$d\nu$ in~$\ndm$. 
	Then the operators~$\nu(\hat{K})$ and~$U \? \nu(\hat{K}) \? U^{-1}$ have the same spectrum.
\end{Lemma}
\begin{proof}
	Applying Proposition~\ref{Proposition spectrum}, we infer that the operators~${\nu(\hat{K}) = (\nu(\hat{K}) \? U^{-1}) \? U}$ and~${U\? \nu(\hat{K}) \? U^{-1} = U \? (\nu(\hat{K}) \? U^{-1})}$ have the same spectrum for any unitary matrix~$U$ in~$\LL(V)$.  
\end{proof}

\begin{Corollary}
	For any negative definite measure~$d\nu \in \ndm$ and arbitrary unitary transformations~$U$ on~$V$ (with respect to~$\prec . \mid . \succ$), 
	\begin{align}\label{(unitary Lagrangian)}
	\L[U \? \nu \? U^{-1}] = \L[\nu] \qquad \text{and} \qquad \Sact(U \? \nu \? U^{-1}) = \Sact(\nu) \:. 
	\end{align}
\end{Corollary}
\begin{proof}
	Introducing the kernel of the fermionic projector by~\eqref{(fermionic projector)} and making use of Definition~\ref{Definition integral negative}, for all~$u,w \in V$ and~$\xi \in \scrM$ we obtain 
	\begin{align*}
	&{\prec u \mid P[U \? \nu \? U^{-1}](\xi) \: w \succ} = {\prec u \mid \int_{\hat{K}} e^{ik\xi} \: d \left(U \? \nu \? U^{-1}\right)(k) \? w \succ} \\
	&\qquad = \int_{\hat{K}} e^{ik\xi} \: {d\prec u \mid U \? \nu(k) \? U^{-1} \? w \succ} = \int_{\hat{K}} e^{ik\xi} \: {d\prec U^{-1} \? u \mid \nu(k) \? U^{-1} \? w \succ} \\
	&\qquad= {\prec U^{-1} \? u \mid \int_{\hat{K}} e^{ik\xi} \: d\nu (k) \? U^{-1} \? w \succ} = {\prec u \mid U \int_{\hat{K}} e^{ik\xi} \: d\nu (k) \? U^{-1} \? w \succ} \\
	&\qquad = {\prec u \mid U \? P[\nu](\xi) \? U^{-1} \? w \succ} \phantom{\int}
	\end{align*}
	for any negative definite measure~$d\nu \in \ndm$ and any unitary matrix~$U$ (with respect to~$\prec . \mid . \succ$). Thus non-degeneracy of the indefinite inner product implies that 
	\begin{align*}
	P[U \? \nu \? U^{-1}] = U \? P[\nu] \? U^{-1} \:.
	\end{align*}
	Henceforth, employing Lemma~\ref{Lemma spectrum}, we deduce that the spectral weight of the closed chain~$A$ 
	is unaffected by unitary similarity, i.e. 
	\begin{align*}
	\left|A[U \? \nu \? U^{-1}](\xi)\right| = \left|U \? A[\nu](\xi) \? U^{-1}\right| = \left|A[\nu](\xi) \right| \qquad \text{for all $\xi \in \scrM$} \:. 
	\end{align*}
	Analogously, for every~$\xi \in \scrM$ we obtain 
	\begin{align*}
	\left|A[U \? \nu \? U^{-1}](\xi)^2 \right| = \left|\big(U \? A[\nu](\xi) \? U^{-1}\big)^2 \right| = \left|A[\nu](\xi)^2 \right| \:, 
	\end{align*}
	thus implying that
	\begin{align*}
	\L[U \? \nu \? U^{-1}](\xi) = \L[\nu](\xi) \qquad \text{for all~$\xi \in \scrM$}
	\end{align*}
	as well as~$\Sact(U \? \nu \? U^{-1}) = \Sact(\nu)$. This completes the proof. 
\end{proof}

We are now in the position to prove the following result: 
\begin{Lemma}\label{Lemma uniform boundedness}
	Let~$f > 0$ and assume that~$(d\nu^{(k)})_{k \in \N}$ is a sequence in~$\ndm$ such that
	\begin{align*}
	\Tr_V \big(-S\nu^{(k)}(\hat{K})\big) \le f \qquad \text{for all~$k \in \N$} 
	\end{align*}
	(where~$S$ denotes the signature matrix). Then there exists
	a positive constant~$C > 0$ in such a way that~$d\|\nu^{(k)}\| \le C$ for all~$k \in \N$, where~$d\|\cdot\|$ denotes the total variation according to Definition~\ref{Definition variation}. 
\end{Lemma}
\begin{proof}
	For convenience, we fix an arbitrary integer~$k \in \N$ and let~$d\nu = d\nu^{(k)}$. Next, we let~$(\mathfrak{e}_i)_{i=1, \ldots,2n}$ be a pseudo-orthonormal basis of~$V$ with signature matrix~$S$ such that~\eqref{(3.13)} is satisfied. Then~${d\prec \mathfrak{e}_i \mid \nu \: \mathfrak{e}_j \succ}$ is a finite complex measure in~$\mathbf{M}_{\C}(\hat{K})$ for every~$i,j \in \{1, \ldots, 2n \}$ according to Definition~\ref{Definition Operator-valued measure}, i.e. 
	\begin{align*}
	d\|{\prec \mathfrak{e}_i \mid \nu \: \mathfrak{e}_j \succ}\| = d|{\prec \mathfrak{e}_i \mid \nu \: \mathfrak{e}_j \succ}|(\hat{K}) < \infty \qquad \text{for all~$i, j = 1, \ldots, 2n$} \:.
	\end{align*}
	Employing the definition of the total variation of complex measures and applying the Schwarz inequality (see e.g.~\cite[Lemma A.13]{langerma} or~\cite[ineq.\ (2.3.9)]{gohberg}), we obtain 
	\begin{align*}
	d\|{\prec \mathfrak{e}_i \mid \nu \: \mathfrak{e}_j \succ}\| &= \sup \sum_{n \in \N} \left|\prec \mathfrak{e}_i \mid \nu(E_n) \: \mathfrak{e}_j \succ \right| = \sup \sum_{n \in \N} \left|\prec \mathfrak{e}_i \mid -\nu(E_n) \: \mathfrak{e}_j \succ \right| \\
	&\le \sup \sum_{n \in \N} \sqrt{\left|\prec \mathfrak{e}_i \mid - \nu(E_n) \: \mathfrak{e}_i \succ \right|} \: \sqrt{ \left|\prec \mathfrak{e}_j \mid - \nu(E_n) \: \mathfrak{e}_j \succ \right|} \:, 
	\end{align*}
	where the supremum is taken over all partitions~$(E_n)_{n \in \N}$ of~$\hat{K}$ (cf.~\cite[Chapter 6]{rudin}). Applying Young's inequality (see e.g.~\cite[\S 1]{alt}), for all~$i,j \in \{1, \ldots, 2n \}$ we arrive at
	\begin{align*}
	d\|{\prec \mathfrak{e}_i \mid \nu \: \mathfrak{e}_j \succ}\| &\le \sup \sum_{n \in \N} \sqrt{\left|\prec \mathfrak{e}_i \mid - \nu(E_n) \: \mathfrak{e}_j \succ \right|} \: \sqrt{ \left|\prec \mathfrak{e}_j \mid - \nu(E_n) \: \mathfrak{e}_j \succ \right|} \\
	&\le \frac{1}{2} \sup \sum_{n \in \N} \left(\left|\prec \mathfrak{e}_i \mid - \nu(E_n) \: \mathfrak{e}_i \succ \right| + \left|\prec \mathfrak{e}_j \mid - \nu(E_n) \: \mathfrak{e}_j \succ \right| \right) \\
	&\le \frac{1}{2} \left[\sup \sum_{n \in \N} \left|\prec \mathfrak{e}_i \mid - \nu(E_n) \: \mathfrak{e}_i \succ \right| + \sup \sum_{n \in \N} \left|\prec \mathfrak{e}_j \mid - \nu(E_n) \: \mathfrak{e}_j \succ \right| \right]\\
	&= \frac{1}{2} \left(d\|{\prec \mathfrak{e}_i \mid \nu \: \mathfrak{e}_i \succ} \| + d\|{\prec \mathfrak{e}_j \mid \nu \: \mathfrak{e}_j \succ}\|\right) \:. 
	\end{align*}
	Due to the fact that~${d\prec \mathfrak{e}_i \mid - \nu \: \mathfrak{e}_i \succ}$ is a positive measure for each~$i \in \{1, \ldots, 2n \}$, the total variation~$d\|{\prec \mathfrak{e}_i \mid \nu \: \mathfrak{e}_j \succ}\| $ is bounded by 
	\begin{align*}
	d\|{\prec \mathfrak{e}_i \mid \nu \: \mathfrak{e}_j \succ}\| \le \sum_{i=1}^{2n} d\|{\prec \mathfrak{e}_i \mid \nu \: \mathfrak{e}_i \succ}\| = \sum_{i=1}^{2n} \prec \mathfrak{e}_i \mid -\nu(\hat{K}) \: \mathfrak{e}_i \succ 
	\end{align*}
	for all~$i,j \in \{1, \ldots, 2n \}$. The last expression can be estimated by 
	\begin{align*}
	\sum_{i=1}^{2n} {\prec \mathfrak{e}_i \mid -\nu(\hat{K}) \? \mathfrak{e}_i \succ} = \sum_{i=1}^{2n} \langle \mathfrak{e}_i \mid -S \nu(\hat{K}) \? \mathfrak{e}_i \rangle = \Tr_V \big(-S\nu(\hat{K}) \big) \le f \:, 
	\end{align*}
	thus completing the proof. 
\end{proof}

In the case that the spectral weight is bounded (in analogy to~\cite[Theorem~6.1]{discrete}), we obtain the following result: 

\begin{Lemma}\label{Lemma spectral weight bounded}
	Let~$f > 0$ and assume that~$(d\nu^{(k)})_{k \in \N}$ is a sequence in~$\ndm$ such that 
	\begin{align*}
	|\nu^{(k)}(\hat{K})| \le f \qquad \text{for all~$k \in \N$}
	\end{align*}
	(where~$|\cdot|$ denotes the spectral weight). Then there is a sequence~$(U_k)_{k \in \N}$ of unitary operators on~$V$ (with respect to~$\prec . \mid . \succ$) as well as a positive constant $C > 0$ such that~$d\|U_k \: \nu^{(k)} \: U_k^{-1} \| \le C$ for all~$k \in \N$ (where~$d\|\cdot\|$ denotes the total variation according to Definition~\ref{Definition variation}). 
\end{Lemma}
\begin{proof}
	The basic idea is to make use of~\cite[Lemma~4.4]{continuum}. For convenience, we fix an arbitrary integer~$k \in \N$ and let~$d\nu = d\nu^{(k)}$. Moreover, let~$(\mathfrak{e}_i)_{i=1, \ldots,2n}$ be a pseudo-orthonormal basis of~$V$ with signature matrix~$S$ such that~\eqref{(3.13)} is satisfied (see for instance~\cite[\S 2.3]{gohberg} or~\cite[\S 3.3]{langerma}). Since~$V$ is a finite-dimensional vector space, all norms on~$\LL(V)$ are equivalent, and one of these norms is given by
	\begin{align}\label{(Spaltensummennorm)}
	\|A\|_1 = \max_{j=1, \ldots, 2n} \sum_{i=1}^{2n} \left|\langle \mathfrak{e}_i \mid A \mathfrak{e}_j \rangle \right|
	\end{align}
	for any~$A \in \LL(V)$, where~$\left|\cdot\right|$ denotes the absolute value. Moreover, for any unitary matrix~$U$ in~$\LL(V)$ (with respect to~$\prec . \mid . \succ$), we may introduce another pseudo-orthonormal basis~$(\mathfrak{f}_j)_{j=1, \ldots, 2n}$ by
	\begin{align}\label{(basis f)}
	\mathfrak{f}_i := U^{-1} \: \mathfrak{e}_i \qquad \text{for all~$i=1, \ldots, 2n$} \:.
	\end{align}
	Making use of~$U^{\ast} = U^{-1}$, for all~$i,j = 1, \ldots, 2n$ we obtain
	\begin{align}\label{(similarity)}
	{d\prec \mathfrak{e}_i \mid U \: \nu \: U^{-1} \mathfrak{e}_j \succ} = {d\prec U^{\ast} \: \mathfrak{e}_i \mid \nu \: U^{-1} \: \mathfrak{e}_j \succ} 
	= {d\prec \mathfrak{f}_i \mid \nu \: \mathfrak{f}_j \succ} \:. 
	\end{align}
	Since~$d\nu$ is a negative definite measure, the operator~$-\nu(\hat{K})$ is positive~\eqref{(4.2')}. 
	Thus in view of~\cite[Lemma~4.4]{continuum}, for any~$\varepsilon > 0$ there is a unitary matrix~$U = U(\varepsilon)$ in~$\LL(V)$ (with respect to~$\prec . \mid . \succ$) so that~$U\: \nu(\hat{K}) \: U^{-1}$ is diagonal, up to an arbitrarily small error term~$\Delta \nu(\hat{K})$ with~$\|\Delta \nu(\hat{K})\|_1 < \varepsilon$. Since~$k \in \N$ is arbitrary, we thus obtain a sequence of negative definite measures~$(U_k \: d\nu^{(k)} \: U_k^{-1})_{k \in \N}$. 
	
	Next, in order to prove that~$(U_k \: d\nu^{(k)} \: U_k^{-1})_{k \in \N}$ is bounded with respect to the total variation defined by~\eqref{(norm of a negative definite measure)}, for each~$k \in \N$ we consider the basis~$(\mathfrak{f}_i)_{i=1, \ldots, 2n}$ given by~\eqref{(basis f)} with respect to the unitary matrix~$U = U_k$. 
	Accordingly, each~${d\prec \mathfrak{f}_i \mid \nu \: \mathfrak{f}_j \succ}$ is a finite complex measure in~$\mathbf{M}_{\C}(\hat{K})$ in view of Definition~\ref{Definition Operator-valued measure}, 
	\begin{align*}
	d\|{\prec \mathfrak{f}_i \mid \nu \: \mathfrak{f}_j \succ}\| = d|{\prec \mathfrak{f}_i \mid \nu \: \mathfrak{f}_j \succ}|(\hat{K}) < \infty \qquad \text{for all~$i, j = 1, \ldots, 2n$} \:.
	\end{align*}
	Employing the definition of the total variation of complex measures and applying the Schwarz inequality in analogy to the proof of Lemma~\ref{Lemma uniform boundedness}, we obtain 
	\begin{align*}
	d\|{\prec \mathfrak{f}_i \mid \nu \: \mathfrak{f}_j \succ}\| 
	\le \sup \sum_{n \in \N} \sqrt{\left|\prec \mathfrak{f}_i \mid - \nu(E_n) \: \mathfrak{f}_i \succ \right|} \: \sqrt{ \left|\prec \mathfrak{f}_j \mid - \nu(E_n) \: \mathfrak{f}_j \succ \right|} \:, 
	\end{align*}
	where the supremum is taken over all partitions~$(E_n)_{n \in \N}$ of~$\hat{K}$ (cf.~\cite[Chapter 6]{rudin}). Applying Young's inequality in analogy to the proof of Lemma~\ref{Lemma uniform boundedness}, we arrive at
	\begin{align*}
	d\|{\prec \mathfrak{f}_i \mid \nu \: \mathfrak{f}_j \succ}\| \le 
	\frac{1}{2} \left(d\|{\prec \mathfrak{f}_i \mid \nu \: \mathfrak{f}_i \succ} \| + d\|{\prec \mathfrak{f}_j \mid \nu \: \mathfrak{f}_j \succ}\|\right) 
	\end{align*}
	for all~$i,j \in \{1, \ldots, 2n \}$. Since~$S = S^{-1}$ and~$U^{\ast} = S^{-1} \: U^{\dagger} \: S$ in view of~\cite[eq.~(4.1.3)]{gohberg} (where~$U^{\dagger}$ denotes the adjoint with respect to~$\langle \, . \, | \, . \, \rangle$ and~$U^{\ast}$ the adjoint with respect to~$\prec . \mid .\succ$), 
	for all~$i=1, \ldots, 2n$ we obtain 
	\begin{align*}
	&d\|{\prec \mathfrak{f}_i \mid -\nu \: \mathfrak{f}_i \succ}\| = {\prec \mathfrak{f}_i \mid - \nu(\hat{K}) \: \mathfrak{f}_i \succ} \le 
	\sum_{i,j=1}^{2n} \left| \prec U^{-1} \: \mathfrak{e}_i \mid \nu(\hat{K}) \: U^{-1} \: \mathfrak{e}_j \succ \right| \\
	&\qquad \le \sum_{i,j=1}^{2n} \left| \prec SU^{\ast} \: SS \mathfrak{e}_i \mid S\nu(\hat{K}) \: U^{-1} \: \mathfrak{e}_j \succ \right|  
	= \sum_{i,j=1}^{2n} \left| \langle U^{\dagger} \: \mathfrak{e}_i \mid \nu(\hat{K}) \: U^{-1} \: \mathfrak{e}_j \rangle \right| \: |s_i| \\
	&\qquad = \sum_{i,j =1}^{2n} \left| \langle \mathfrak{e}_i \mid U\:  \nu(\hat{K}) \: U^{-1} \:\mathfrak{e}_j \rangle \right| \stackrel{\eqref{(Spaltensummennorm)}}{\le} 2n \:  \|U \: \nu(\hat{K}) \: U^{-1} \|_1 \:, 
	\end{align*}
	where we made use of~$S\mathfrak{e}_i = s_i \mathfrak{e}_i$ with~$|s_i| = |\langle \mathfrak{e}_i \mid S \mathfrak{e}_i \rangle | = 1$ for all~$i=1, \ldots, 2n$ and employed the fact that~${d \prec \mathfrak{f}_i \mid -\nu \: \mathfrak{f}_i \succ}$ is a positive measure for any~$i \in \{1, \ldots, 2n\}$. 
	
	Taken the previous results together, by~\eqref{(similarity)} we obtain the inequality 
	\begin{align}\label{(inequality)}
	d\|{\prec \mathfrak{e}_i \mid U \: \nu \: U^{-1} \: \mathfrak{e}_i \succ}\| = d\|{\prec \mathfrak{f}_i \mid -\nu \: \mathfrak{f}_i \succ}\| \le 2n \: \|U \: \nu(\hat{K}) \: U^{-1} \|_1
	\end{align}
	for all~$i = 1, \ldots, 2n$. Thus it only remains to find an upper bound for~$\|U \: \nu(\hat{K}) \: U^{-1} \|_1$ in terms of~$f$ by establishing a connection to the spectral weight~$|\nu(\hat{K})|$. To this end we exploit the fact that~${U \? \nu(\hat{K}) \? U^{-1}}$ is diagonal according to~\cite[Lemma~4.4]{continuum}, up to an arbitrarily small error term~$\Delta \nu(\hat{K})$, 
	\begin{align*}
	U \? \nu(\hat{K}) \? U^{-1} = \diag\big(\tilde{\lambda}_1(U), \ldots, \tilde{\lambda}_{2n}(U) \big) + \Delta \nu(\hat{K}) \:. 
	\end{align*}
	Denoting 
	the eigenvalues of~${U \? \nu(\hat{K}) \? U^{-1}}$ by~$\lambda_i(U)$ for all~$i = 1, \ldots, 2n$, by choosing the error term~$\Delta \nu(\hat{K})$ sufficiently small we can arrange that
	\begin{align*}
	\sum_{i=1}^{2n} |\tilde{\lambda}_i(U) - \lambda_i(U)| < \varepsilon \qquad \text{for any~$\varepsilon > 0$} \:. 
	\end{align*}
	Since the off-diagonal elements~$\|\Delta \nu(\hat{K})\|_1 < \varepsilon$ are arbitrarily small, we thus obtain
	\begin{align*}
	\|U \? \nu(\hat{K}) \? U^{-1} \|_1 \le \|\diag(\tilde{\lambda}_1(U), \ldots, \tilde{\lambda}_{2n}(U)) \|_1 + \|\Delta \nu(\hat{K})\|_1 \le \sum_{i=1}^{2n} |\lambda_i(U)| + 2\varepsilon \:.
	\end{align*}
	Applying Lemma~\ref{Lemma spectrum}, we conclude that~$|\nu(\hat{K})| = |U \? \nu(\hat{K}) \? U^{-1}|$ (where~$|\cdot|$ denotes the spectral weight). 
	Choosing~$\varepsilon < 1/2$, we arrive at
	\begin{align*}
	\|U \? \nu(\hat{K}) \? U^{-1} \|_1 \le |\nu(\hat{K})| + 1 \le f+ 1 \:. 
	\end{align*}
	Hence in view of Definition~\ref{Definition variation} and~\eqref{(inequality)}, we finally obtain
	\begin{align*}
	d\| U_k \: \nu^{(k)} \: U_k^{-1} \| = \sum_{i,j = 1}^{2n} d\|{\prec \mathfrak{e}_i \mid U_k \: \nu^{(k)} \: U_k^{-1} \: \mathfrak{e}_j \succ} \| \le (2n)^3 \: (f+1) =: C \:. 
	\end{align*}
	This completes the proof. 
\end{proof}

The major simplification when restricting attention to compact subsets is that any minimizing sequence is uniformly tight a~priori. As a consequence, we may apply Prohorov's theorem to each component, thereby obtaining the desired minimizer. 

\subsection{Preparatory Result}\label{S preparatory} 
Given a sequence of negative definite measures which is bounded and uniformly tight, 
we employ Prohorov's theorem to prove that a subsequence thereof converges weakly (see Definition~\ref{Definition weak convergence}) to a negative definite measure:

\begin{Lemma}\label{Lemma negative}
	Let~$(d\nu_k)_{k \in \N}$ be a sequence of negative definite measures in~$\ndm$ with the following properties:  
	\begin{itemize}[leftmargin=2em]
		\item[\rm{(a)}] 
		There is a constant~$C > 0$ such that~$d|\nu_k|(\hat{K}) \le C$ for all~$k \in \N$. 
		\item[\rm{(b)}] The sequence~$(d\nu_k)_{k \in \N}$ is uniformly tight in the sense that, for every~$\varepsilon > 0$, there is a compact subset~$K_{\varepsilon} \subset \hat{K}$ such that~$d|\nu_k|(\hat{K} \setminus K_{\varepsilon}) < \varepsilon$ for all~$k \in \N$.  
	\end{itemize}
	Then a subsequence of~$(d\nu_k)_{k \in \N}$ converges weakly to some negative definite measure~$d\nu$.
\end{Lemma}
\begin{proof}
	The main idea is to apply Prohorov's theorem. More precisely, let~$(\mathfrak{e}_i)_{i=1, \ldots, 2n}$ be a pseudo-orthonormal basis of~$V$ satisfying~\eqref{(3.13)}, and for every~$k \in \N$ we denote by~$d|\nu_k|$ the corresponding variation of~$d\nu_k$ according to Definition~\ref{Definition variation}. Decomposing the complex measure~$d\prec \mathfrak{e}_i \mid -\nu_k \: \mathfrak{e}_j \succ$ 
	into its real and imaginary part, 
	\begin{align*}
	{d\prec \mathfrak{e}_i \mid -\nu_k \: \mathfrak{e}_j \succ} = \re d\prec \mathfrak{e}_i \mid -\nu_k \: \mathfrak{e}_j \succ + \: i \im d\prec \mathfrak{e}_i \mid -\nu_k \: \mathfrak{e}_j \succ \:, 
	\end{align*}
	and introducing the (positive) measures
	\begin{align*}
	d\Re_{[i,j], k}^{\pm} := \re \: d\prec \mathfrak{e}_i \mid -\nu_k \: \mathfrak{e}_j \succ^{\pm} \qquad \text{and} \qquad d\Im_{[i,j], k}^{\pm} := \im \: d\prec \mathfrak{e}_i \mid -\nu_k \: \mathfrak{e}_j \succ^{\pm} 
	\end{align*}
	by applying the Jordan decomposition~\cite[\S 29]{halmosmt}, 
	we arrive at
	\begin{align*}
	{d\prec \mathfrak{e}_i \mid -\nu_k \: \mathfrak{e}_j \succ} = d\Re_{[i,j], k}^{+} - d\Re_{[i,j], k}^{-} + \: i \: d\Im_{[i,j], k}^{+} - \: i \: d\Im_{[i,j], k}^{-}
	\end{align*}
	for all~$i,j \in \{ 1, \ldots, 2n \}$ and each~$k \in \N$. Then the conditions~(a) and~(b) imply that the sequences~$(d\Re_{[i,j], k}^{\pm})_{k \in \N}$ and~$(d\Im_{[i,j], k}^{\pm})_{k \in \N}$ are bounded and uniformly tight for all~$i,j = 1, \ldots, 2n$. Iteratively applying Prohorov's theorem, we deduce that~$(d\nu_k)_{k \in \N}$ contains a subsequence (which for convenience we again denote by~$(d\nu_k)_{k \in \N}$) such that the corresponding sequences~$(d\Re_{[i,j], k}^{\pm})_{k \in \N}$ and~$(d\Im_{[i,j], k}^{\pm})_{k \in \N}$ weakly converge to (positive) measures~$d\Re_{[i,j]}^{\pm}$ and~$d\Im_{[i,j]}^{\pm}$, respectively, i.e. 
	\begin{align*}
	d\Re_{[i,j], k}^{\pm} \rightharpoonup d\Re_{[i,j]}^{\pm} \qquad \text{and} \qquad d\Im_{[i,j], k}^{\pm} \rightharpoonup d\Im_{[i,j]}^{\pm}
	\end{align*}
	for all~$i,j \in \{ 1, \ldots, 2n \}$ and every~$k \in \N$. Introducing the measures
	\begin{align*}
	d\nu_{i,j} := d\Re_{[i,j]}^{+} - d\Re_{[i,j]}^{-} + \: i \: d\Im_{[i,j]}^{+} - \: i \: d\Im_{[i,j]}^{-} \qquad \text{for all~$i,j \in \{1, \ldots, 2n\}$} \:, 
	\end{align*}
	for every~$f \in C_b(\hat{K})$ we obtain weak convergence
	\begin{align*}
	\lim_{k \to \infty} \int_{\hat{K}} f \: {d\prec \mathfrak{e}_i \mid - \nu_k \: \mathfrak{e}_j \succ} = \int_{\hat{K}} f \: d\nu_{i,j} \qquad \text{for all~$i,j \in \{1, \ldots, 2n\}$} \:.  
	\end{align*}
	Following the proof of Lemma~\ref{Lemma OVM Banach space}, we introduce the operator-valued measure~$d\nu$ for every~$\Omega \in \B(\hat{K})$ by
	\begin{align*}
	\nu(\Omega) := \begin{pmatrix}
	&\nu_{1,1}(\Omega) &\cdots &\nu_{1, 2n}(\Omega) \\
	&\vdots & \ddots&\vdots \\
	&\nu_{n,1}(\Omega) &\cdots &\nu_{n,2n}(\Omega) \\
	&-\nu_{n+1,1}(\Omega) &\cdots &-\nu_{n+1,2n}(\Omega) \\
	&\vdots & \ddots&\vdots\\
	&-\nu_{1,2n}(\Omega) & \cdots &-\nu_{2n,2n}(\Omega)
	\end{pmatrix} \in \LL(V) \:. 
	\end{align*}
	The measure~$d\nu$ has the property that, for all~$i,j \in \{1, \ldots, 2n\}$, 
	\begin{align*}
	{d\prec \mathfrak{e}_i \mid \nu \: \mathfrak{e}_j \succ} = d\langle \mathfrak{e}_i \mid S \: \nu \: \mathfrak{e}_j \rangle = d\nu_{i,j} \in \mathbf{M}_{\C}(\hat{K})
	\end{align*}
	is a complex measure. For elements~$u= \sum_{m=1}^{2n} \alpha_j(u) \: \mathfrak{e}_j$ and~$v = \sum_{m=1}^{2n} \alpha_j(v) \: \mathfrak{e}_j$ in~$V$,  
	by linearity we conclude that~${d\prec u \mid \nu \: v \succ} \in \mathbf{M}_{\C}(\hat{K})$ for all~$u,v \in V$. Hence~$d\nu$ is an operator-valued measure in the sense of Definition~\ref{Definition Operator-valued measure}, and by linearity we arrive at 
	\begin{align*}
	&\lim_{k \to \infty} \int_{\hat{K}} f \: {d\prec u \mid - \nu_k \: v \succ} = \lim_{k \to \infty} \sum_{\ell,m=1}^{2n} \overline{\alpha_{\ell}(u)} \: \alpha_m(v) \: \int_{\hat{K}} f \: {d\prec \mathfrak{e}_{\ell} \mid - \nu_k \: \mathfrak{e}_m \succ} \\
	&\qquad = \sum_{\ell,m=1}^{2n} \overline{\alpha_{\ell}(u)} \: \alpha_m(v) \: \int_{\hat{K}} f \: {d\prec \mathfrak{e}_{\ell} \mid - \nu \: \mathfrak{e}_m \succ} = \int_{\hat{K}} f \: {d\prec u \mid - \nu \: v \succ} 
	\end{align*}
	for all~$f \in C_b(\hat{K})$ and~$u, v \in V$. This yields weak convergence~$d\nu_k \rightharpoonup d\nu$ of operator-valued measures in the sense of Definition~\ref{Definition weak convergence}. In particular,~$d\|\nu\| < \infty$.

	It remains to show that~$d\nu$ is indeed negative definite. To this end, we need to prove that~${d\prec u \mid - \nu \: u \succ}$ is a positive measure for all~$u \in V$. We point out that, by assumption, the measures~$d\prec u \mid -\nu_k \: u \succ$ are positive for each~$u \in V$ and all~$k \in \N$. Assume now, for some~$u \in V$, that~$d\mu_u := d\prec u \mid -\nu \: u \succ$ is a signed measure with
	\begin{align*}
	d\mu_u = d\mu_u^+ - d\mu_u^- 
	\end{align*}
	such that~$d\mu_u^{-}$ is non-zero. In this case, there is~$\Omega \in \B(\hat{K})$ with the property that
	\begin{align*}
	\mu_u^+(\Omega) < \mu_u^-(\Omega)
	\end{align*}
	(assuming conversely that~$\mu_u^+(\Omega) \ge \mu_u^-(\Omega)$ for all~$\Omega \in \B(\hat{K})$, then the  measure~$d\mu_u$ is non-negative, implying that~$d\mu_u^-  = 0$). In virtue of Ulam's theorem we know that both measures~$d\mu_u^{\pm}$ are regular on~$\hat{K}$. As a consequence, there is an open set~$U \supset \Omega$ and a compact set~$K \subset \Omega$ such that~$\mu_u^+(U) < \mu_u^-(K)$. 
	Hence a partition of unity yields a function~$f \in C_c(U;[0,1])$ with~$\supp f \subset U$ and~$f|_K \equiv 1$, thus giving rise to the contradiction
	\begin{align*}
	0 \le \lim_{k \to \infty} \int_{\hat{K}} f \: {d\prec u \mid - \nu_k \: u \succ} = \int_{\hat{K}} f \: {d\prec u \mid - \nu \: u \succ} \le \mu_u^+(U) - \mu_u^-(K) < 0 \:. 
	\end{align*}
	This completes the proof. 
\end{proof}

\subsection{Proof of Existence Theorem}\label{S existence} 
In order for proving Theorem~\ref{Theorem minimizer}, 
we require some more preparatory results. The proof of Theorem~\ref{Theorem minimizer} will be completed towards the end of this subsection. 
To begin with, let us state the following proposition. 

\begin{Prp}\label{Proposition Fatou}
	Let~$(d\nu^{(k)})_{j \in \N}$ be a sequence of negative definite measures in~$\ndm$ which converges weakly to some negative definite measure~$d\nu \in \ndm$. Then
	\begin{align*}
	\lim_{j \to \infty} \L[\nu^{(j)}](\xi) = \L[\nu](\xi) \qquad \text{for all~$\xi \in \scrM$}
	\end{align*}
	and~$$\Sact(\nu) \le \liminf_{j \to \infty} \Sact(\nu^{(j)}) \:. $$
\end{Prp}
\begin{proof}
	Let us first consider the behavior of the kernel of the fermionic projector and the closed chain. 
	For convenience, we introduce the notation~$P_j(\xi) := P[\nu^{(j)}](\xi)$ as well as~$A_j(\xi) := A[\nu^{(j)}](\xi)$ for all~$j \in \N$ and arbitrary~$\xi \in \scrM$. Then weak convergence (see Definition~\ref{Definition weak convergence} and the remark thereafter) implies that
	\begin{align*}
	\lim_{j \to \infty} {\prec u \mid P_j(\xi) \: v \succ} &= \lim_{j \to \infty} \int_{\hat{K}} e^{ik\xi} \: {d\prec u \mid \nu^{(j)}(\xi) \: v \succ} \\
	&= \int_{\hat{K}} e^{ik\xi} \: {d\prec u \mid \nu(\xi) \: v \succ} = {\prec u \mid P[\nu](\xi) \: v \succ}
	\end{align*}
	for all~$u,v \in V$ and arbitrary~$\xi \in \scrM$. 
	Given a pseudo-orthonormal basis~$(\mathfrak{e}_i)_{i = 1, \ldots, 2n}$ of~$V$ satisfying~\eqref{(3.13)}, we thus obtain
	\begin{align*}
	\lim_{j \to \infty} \langle \mathfrak{e}_{\alpha} \mid P_j(\xi) \? \mathfrak{e}_{\beta} \rangle = \lim_{j \to \infty} {\prec S\mathfrak{e}_{\alpha} \mid P_j(\xi) \? \mathfrak{e}_{\beta} \succ} = {\prec S\mathfrak{e}_{\alpha} \mid P[\nu](\xi) \? \mathfrak{e}_{\beta} \succ} = \langle \mathfrak{e}_{\alpha} \mid P_j(\xi) \? \mathfrak{e}_{\beta} \rangle 
	\end{align*}
	for all~$\alpha, \beta \in \{1, \ldots, 2n \}$ and arbitrary~$\xi \in \scrM$. From this we deduce that
	\begin{align*}
	\lim_{j \to \infty} (A_j(\xi))_{\alpha, \beta} &= \lim_{j \to \infty} \big(P_j(\xi) \? P_j(-\xi)\big)_{\alpha, \beta} = \big(P[\nu](\xi) \? P[\nu](-\xi)\big)_{\alpha, \beta} = (A[\nu](\xi))_{\alpha, \beta}		
	\end{align*}
	for all~$\alpha, \beta \in \{1, \ldots, 2n \}$ and arbitrary~$\xi \in \scrM$. By continuity of the spectral weight, 
	\begin{align*}
	\lim_{j \to \infty} \L[\nu^{(j)}](\xi) = \L[\nu](\xi) \qquad \text{for all~$\xi \in \scrM$} \:.
	\end{align*}
	The second statement follows from Fatou's lemma (see e.g.~\cite[Theorem~16.4]{jost-postmodern}), 
	\begin{align*}
	\Sact(\nu) = \int_{\scrM} \L[\nu](\xi) \: d\mu(\xi) = \int_{\scrM} \liminf_{j \to \infty} \L[\nu^{(j)}](\xi) \: d\mu(\xi) \le \liminf_{j \to \infty} \int_{\scrM} \L[\nu^{(j)}](\xi) \: d\mu(\xi) \:.
	\end{align*}
	This completes the proof. 
\end{proof}

\begin{Prp}\label{Proposition Tr}
	Let~$(d\nu^{(j)})_{j \in \N}$ be a sequence of negative definite measures in~$\ndm$ which converges weakly to some negative definite measure~$d\nu \in \ndm$. Then 
	\begin{align*}
	\lim_{j \to \infty} \Tr_V(\nu^{(j)}(\hat{K})) = \Tr_V(\nu(\hat{K}))
	\end{align*}
	as well as
	\begin{align*}
	\lim_{j \to \infty} \Tr_V(-S\nu^{(j)}(\hat{K})) = \Tr_V(-S\nu(\hat{K}))  \qquad \text{and} \qquad \lim_{j \to \infty} |\nu^{(j)}(\hat{K})| = |\nu(\hat{K})| \:. 
	\end{align*}
\end{Prp}

For proving the last assertion, we require the next lemma: 
\begin{Lemma}\label{Lemma Kato}
	Let~$W$ be a finite-dimensional vector space and let~$T \in \LL(W)$. Then for any sequence~$(T_n)_{n \in \N}$ of operators in~$\LL(W)$ with~$\|T_n - T\| \to 0$ as~$n \to \infty$ (where~$\|. \|$ denotes any norm on~$W$), the eigenvalues of~$T_n$ converge to those of~$T$. 
\end{Lemma}
\begin{proof}
	See~\cite[Chapter~II, \S 5-1]{kato}. 
\end{proof}

\begin{proof}[Proof of Proposition~\ref{Proposition Tr}]
	By weak convergence, the first two equalities can be verified as follows:
	\begin{align*}
	\lim_{j \to \infty}  \Tr_V(\nu^{(j)}(\hat{K})) &= \lim_{j \to \infty} \sum_{\alpha = 1}^{2n} \langle \mathfrak{e}_{\alpha} \mid \nu^{(j)}(\hat{K}) \? \mathfrak{e}_{\alpha} \rangle = \lim_{j \to \infty} \sum_{\alpha = 1}^{2n} \int_{\hat{K}} {d\prec S\mathfrak{e}_{\alpha} \mid \nu^{(j)}(p) \? \mathfrak{e}_{\alpha} \succ} \\
	&= \sum_{\alpha = 1}^{2n} \int_{\hat{K}} {d\prec S\mathfrak{e}_{\alpha} \mid \nu(p) \? \mathfrak{e}_{\alpha} \succ} = \sum_{\alpha = 1}^{2n} \langle \mathfrak{e}_{\alpha} \mid \nu(\hat{K}) \? \mathfrak{e}_{\alpha} \rangle = \Tr_V(\nu(\hat{K})) \:, 
	\end{align*}
	and analogously 
	\begin{align*}
	\lim_{j \to \infty}  \Tr_V(-S\nu^{(j)}(\hat{K})) 
	= \Tr_V(-S\nu(\hat{K})) \:. 
	\end{align*}
	In order to prove the remaining equality, we essentially make use of the fact that the spectral weight is continuous. More precisely, by continuity of the absolute value and weak convergence we obtain 
	\begin{align*}
	&\lim_{j \to \infty} \|\nu^{(j)}(\hat{K})- \nu(\hat{K})\|_{1} \le \lim_{j \to \infty} \sum_{\alpha, \beta =1}^{2n} \left|\langle \mathfrak{e}_{\alpha} \mid \big(\nu^{(j)}(\hat{K})- \nu(\hat{K})\big) \? \mathfrak{e}_{\beta} \rangle\right|  \\
	&\qquad = \lim_{j \to \infty} \sum_{\alpha, \beta =1}^{2n} \left|\int_{\hat{K}} d\prec S\mathfrak{e}_{\alpha} \mid \nu^{(j)}(k) \? \mathfrak{e}_{\beta} \succ - \int_{\hat{K}} d\prec S \mathfrak{e}_{\alpha} \mid \nu(k) \? \mathfrak{e}_{\beta} \succ\right| = 0
	\end{align*}
	(where~$\|.\|_1$ is given by~\eqref{(Spaltensummennorm)}). 
	Denoting the eigenvalues of~$\nu(\hat{K})$ by~$(\lambda_i)_{i=1, \ldots, 2n}$ and those of~$\nu^{(j)}(\hat{K})$ for every~$j \in \N$ by~$(\lambda_i^{(j)})_{i=1, \ldots, 2n}$, by applying Lemma~\ref{Lemma Kato} together with the inverse triangle inequality we thus arrive at 
	\begin{align*}
	\lim_{j \to \infty} \left||\nu^{(j)}(\hat{K})| - |\nu(\hat{K})|\right| \le \lim_{j \to \infty} \sum_{i=1}^{2n} \left| |\lambda^{(j)}_i| - |\lambda_i| \right| \le \lim_{j \to \infty} \sum_{i=1}^{2n} | \lambda^{(j)}_i - \lambda_i | = 0 \:.
	\end{align*}
	This completes the proof. 
\end{proof}

After these preliminaries we are finally in the position to prove Theorem~\ref{Theorem minimizer}.

\begin{proof}[Proof of Theorem~\ref{Theorem minimizer}]
	Let us first assume that the side conditions~\eqref{(constraints spectral)} are satisfied. In this case, 
	Lemma~\ref{Lemma spectral weight bounded} yields a sequence of unitary operators~$(U_j)_{j \in \N}$ in~$\LL(V)$ (with respect to~$\prec . \mid . \succ$) as well as a constant~$C > 0$ such that 
	\begin{align*}
	d\|U_j \? \nu^{(j)} \? U_j^{-1}\| \le C \qquad \text{for all~$j \in \N$} \:. 
	\end{align*}
	Since~$\hat{K} \subset \hscrM$ is compact, the sequence of measures~$(d\nu^{(j)})_{j \in \N}$ is uniformly tight.  
	As a consequence, we may apply Lemma~\ref{Lemma negative} in order to conclude that a subsequence of~$(U_j \? d\nu^{(j)} \? U_j^{-1})_{j \in \N}$ converges weakly to some negative definite measure~$d\nu \in \ndm$, 
	\begin{align*}
	d\tilde{\nu}^{(j_k)} := U_{j_k} \? d\nu^{(j_k)} \? U_{j_k}^{-1} \rightharpoonup d\nu \qquad \text{weakly} \:.
	\end{align*}
	Making use of~\eqref{(unitary Lagrangian)}, from Proposition~\ref{Proposition Fatou} we deduce that 
	\begin{align*}
	\Sact(\nu) \le \lim_{k \to \infty} \Sact(\tilde{\nu}^{(j_k)}) = \lim_{k \to \infty} \Sact({\nu}^{(j_k)}) \:. 
	\end{align*}
	In the case that the constraints~\eqref{(constraints)} are imposed, the above arguments remain valid by applying Lemma~\ref{Lemma uniform boundedness} instead of Lemma~\ref{Lemma spectral weight bounded} and choosing~$U_j = \Id_V$ for all~$j \in\N$. 
	
	Thus it only remains to prove that the measure~$d\nu$ satisfies the conditions~\eqref{(limiting conditions)} or~\eqref{(limiting conditions spectral)}, respectively. In both cases, this follows readily from Proposition~\ref{Proposition Tr}. In particular, the limit measure~$d\nu$ is non-trivial, which completes the proof. 
\end{proof}

As worked out in the next subsection, Theorem~\ref{Theorem minimizer} also holds in the case that the side conditions~\eqref{(f)} and~\eqref{(swf)} are replaced by a boundedness constraint in the fashion of~\cite[Section~4]{continuum}.

\subsection{Imposing a Boundedness Constraint}\label{S boundedness constraint}
Let us finally establish a connection to the boundedness constraint as considered in~\cite[Section~4]{continuum} (which originally was proposed in~\cite[eq.~(3.5.10)]{pfp} as a constraint for the causal action principle). In the homogeneous setting, for any operator-valued measure~$d\omega \in \ovm$ we introduce the mapping~$\mathfrak{t}[\omega] : \scrM \to \R_0^+$ by  
\begin{align*}
\mathfrak{t}[\omega](\xi) := \left|A[{\omega}](\xi)\right|^2 \qquad \text{for all~$\xi \in \scrM$} \:. 
\end{align*} 
We then define the functional~$\T : \ovm \to \R_0^+ \cup \{+ \infty \}$ by 
\begin{align*}
\T(\omega) := \int_{\scrM} \mathfrak{t}[{\omega}](\xi)\: {d}\mu(\xi) =  \int_{\scrM} \left|A[{\omega}](\xi)\right|^2\: {d}\mu(\xi) \:.
\end{align*}
Given~$C > 0$, the corresponding \emph{boundedness constraint} reads 
\begin{align}\label{(boundedness constraint)}
\T(\omega) \le C \:.
\end{align} 
In analogy to Theorem~\ref{Theorem minimizer} we then obtain the following existence result: 

\begin{Thm}\label{Theorem boundedness}
	Assume that~$(d\nu^{(j)})_{j \in \N}$ is a minimizing sequence of negative definite measures in~$\ndm$ for the causal variational principle~\eqref{(causal action principle homogeneous)} with respect to the side conditions~\eqref{(c)} and~\eqref{(boundedness constraint)} for some positive constants~$c, C > 0$. Then there exists a sequence of unitary operators~$(U_j)_{j \in \N}$ on $V$ (with respect to $\prec . \mid . \succ$) as well as a  subsequence~$(d\nu^{(j_k)})_{k \in \N}$ such that the sequence~$(U_{j_k} \: d\nu^{(j_k)} \: U_{j_k}^{-1})_{k \in \N}$ converges weakly to some non-trivial negative definite measure~$d\nu \not= 0$. Moreover, 
	\begin{align*}
	\Sact(\nu) \le \liminf_{k \to \infty} \Sact(\nu^{(j_k)}) \:,
	\end{align*}
	and the limit measure~$d\nu \in \ndm$ satisfies the side conditions
	\begin{align*}
	\Tr_V(\nu(\hat{K})) = c \qquad \text{and} \qquad \T(\nu) \le C \:. 
	\end{align*}
	In particular, the limit measure~$d\nu$ is a non-trivial minimizer of the causal variational principle~\eqref{(causal action principle homogeneous)} with respect to the side conditions~\eqref{(c)} and~\eqref{(boundedness constraint)}.
\end{Thm}

For the proof of Theorem~\ref{Theorem boundedness} we make use of the following result:

\begin{Prp}\label{Proposition boundedness}
	Whenever~$d\nu \in \ndm$ is a negative definite measure satisfying the boundedness constraint~\eqref{(boundedness constraint)}, it satisfies condition~\eqref{(swf)} for some constant~$f > 0$. 
\end{Prp}
\begin{proof}
	Let us first note that~$\mathfrak{t}[\nu] \in L^1_{\textup{loc}}(\scrM)$ whenever~$d\nu \in \ndm$ satisfies~\eqref{(boundedness constraint)}. In analogy to~\cite[Section~3.4]{folland}, for every~$f \in L^1_{\textup{loc}}(\scrM)$ 
	we then introduce 
	\begin{align*}
	G_r \? f(x) := \frac{1}{\mu(B_r(x))} \int_{B_r(x)} f(y) \: d\mu(y) \qquad \text{for all~$x \in \scrM$ and~$r > 0$} \:,
	\end{align*}
	and by virtue of~\cite[Theorem~3.18]{folland} we know that
	\begin{align*}
	\lim_{r \to 0} G_r\? f(x) =f(x) \qquad \text{for almost every~$x \in \scrM$} \:. 
	\end{align*}
	Since~$\mu(B_{\varepsilon}(x)) > 0$ for every~$x \in \scrM$ and arbitrary~$\varepsilon > 0$, continuity of~$A[\nu]$ yields the existence of~$x_0 \in \scrM$ such that 
	\begin{align}\label{(C)}
	|A[\nu](0)|^2 < \mathfrak{t}[\nu](x_0) + 1 = \lim_{\varepsilon \searrow 0} \int_{B_{\varepsilon}(x_0)} |A[\nu](\xi)|^2 \: d\mu(\xi) + 1 \le C +1 \:. 
	\end{align}
	
	We now apply~\eqref{(C)} in order to prove that~$|\nu(\hat{K})| < f$ for some constant~$f > 0$. To this end, we essentially employ~\cite[Lemma~4.4]{continuum}. More precisely, for any negative definite measure~$d\nu$ and arbitrary~$\varepsilon > 0$, there is a unitary operator~$U \in \LL(V)$ (with respect to~$\prec. \mid . \succ$) such that
	\begin{align*}
	U \? \nu(\hat{K}) \? U^{-1} = - \diag(\tilde{\lambda}_1, \ldots, \tilde{\lambda}_{2n}) + \Delta \nu(\hat{K}) \:,
	\end{align*}
	where the real parameters~$\tilde{\lambda}_i$ ($i=1, \ldots, 2n$) are ordered according to~\cite[eq.~(2.6)]{continuum}, and~$\|\Delta \nu(\hat{K})\| < \varepsilon$. Denoting by~$\{., . \}$ the anti-commutator, we thus obtain 
	\begin{align*}
	U \? A[\nu](0) \? U^{-1} &= \big(U \? \nu(\hat{K}) \? U^{-1}\big)^2 \\
	&= \diag\big(\tilde{\lambda}_1^2, \ldots, \tilde{\lambda}_{2n}^2 \big) - \left\{ \diag(\tilde{\lambda}_1, \ldots, \tilde{\lambda}_{2n}), \Delta \nu(\hat{K}) \right\} + \Delta \nu(\hat{K})^2 \:. 
	\end{align*}
	Since~$\|\nu(\hat{K})\| < \infty$, the absolute values of~$\tilde{\lambda}_i$ are bounded for all~$i = 1, \ldots, 2n$; from this we conclude that the spectrum of~$\diag\big(\tilde{\lambda}_1^2, \ldots, \tilde{\lambda}_{2n}^2 \big)$ coincides with the spectrum of $A[\nu](0)$, up to an arbitrarily small error term (where we applied the fact that the spectra of $A[\nu](0)$ and $U \? A[\nu](0) \? U^{-1}$ coincide according to Lemma~\ref{Lemma spectrum}). In a similar fashion, one can show that the spectra of~$\nu(\hat{K})$ and~$- \diag(\tilde{\lambda}_1, \ldots, \tilde{\lambda}_{2n})$ coincide, up to an arbitrarily small error term. Neglecting the error terms in what follows, we thus can arrange that
	\begin{align*}
	|\nu(\hat{K})| \le 2 \sum_{i=1}^{2n} |\tilde{\lambda}_i| \qquad \text{and} \qquad \sum_{i=1}^{2n} \tilde{\lambda}_i^2 \le 2 |A[\nu](0)| \:. 
	\end{align*}
	Employing Jensen's inequality, we conclude that
	\begin{align*}
	|\nu(\hat{K})|^2 \le 4 \left(\sum_{i=1}^{2n} |\tilde{\lambda}_i|\right)^2 \le 8n \sum_{i=1}^{2n} |\tilde{\lambda}_i|^2 \le 16n |A[\nu](0)| \:. 
	\end{align*}
	Applying~\eqref{(C)}, the boundedness constraint gives rise to the desired estimate
	\begin{align*}
	|\nu(\hat{K})| < 4 \? \sqrt{n \? (C+1)} =: f \:, 
	\end{align*}
	which completes the proof. 
\end{proof}

This allows us to prove Theorem~\ref{Theorem boundedness}: 

\begin{proof}[Proof of Theorem~\ref{Theorem boundedness}] 
	We basically combine Proposition~\ref{Proposition boundedness} and Theorem~\ref{Theorem minimizer}.  
	To this end let~$(d\nu^{(j)})_{j \in \N}$ be a minimizing sequence of negative definite measures
	which satisfies
	the side conditions~\eqref{(c)} and~\eqref{(boundedness constraint)} for some positive constants~$c, C > 0$. Then by Proposition~\ref{Proposition boundedness}, there exists~$f > 0$ in such a way that condition~\eqref{(swf)} is satisfied for every~$j \in \N$. As a consequence, 
	according to Theorem~\ref{Theorem minimizer}, there is a sequence of unitary operators~$(U_j)_{j \in \N}$ in~$\LL(V)$ (with respect to~$\prec . \mid . \succ$) such that the sequence~$(U_j \? d\nu^{(j)} \? U_j^{-1})_{j \in \N}$ contains a subsequence (which for simplicity we again denote by~$(U_j \? d\nu^{(j)} \? U_j^{-1})_{j \in \N}$) with the property that it converges weakly to some limit measure~$d\nu \in \ndm$. Applying Fatou's lemma one can show that
	\begin{align*}
	\Sact(\nu) \le \liminf_{j \to \infty} \Sact(\nu^{(j)}) \qquad \text{and} \qquad \T(\nu) \le \liminf_{j \to \infty} \T(\nu^{(j)}) \:. 
	\end{align*}
	By virtue of Proposition~\ref{Proposition Tr} we conclude that~$d\nu$ satisfies condition~\eqref{(c)}, thus implying that~$d\nu \not= 0$ is non-zero. This completes the proof. 
\end{proof}

Thus for compact subsets of momentum space, Theorem~\ref{Theorem boundedness} gives an alternative proof of~\cite[Theorem~4.2]{continuum}.

\appendix

\section{Justifying the Side Conditions}\label{Appendix justification}
This appendix is devoted to justify and explain the side conditions~\eqref{(c)}--\eqref{(swf)}. 
Apart from excluding trivial minimizers in a quite simple way, the following reasoning provides a strong argument for imposing condition~\eqref{(c)}.  
Given a causal fermion system~$(\H, \F, d\rho)$, the so-called local trace~$\tr(x)$ defined by
\begin{align*}
\tr(x) = \Tr_{S_x}\big(P(x,x) \big) \qquad \text{for all~$x \in \supp d\rho$} 
\end{align*}
is constant on~$\supp d\rho$ whenever the measure~$d\rho$ is a minimizer of the causal action principle (for details see~\cite[\S 1.1.3, Proposition~1.4.1 and Section~2.5]{cfs}). 
Considering homogeneous causal fermion systems, this suggests to impose that 
\begin{align*}
\Tr_V(\nu(\hat{K})) = \Tr_V \left(\int_{\hat{K}} d\nu(k)\right) = \Tr_V(P(0)) = \Tr_V(P(x,x)) = c \qquad \text{for all~$x \in \scrM$} \:, 
\end{align*}
thus motivating the side condition~\eqref{(c)}. Following the arguments in~\cite[\S 1.4.1]{cfs}, we shall always assume that~$c \not= 0$, thereby excluding trivial minimizers. 
Let us briefly explain why the quantity~$\Tr_V(P(0)) = \Tr_V(\nu(\hat{K}))$ in~\eqref{(c)} is also referred to as \emph{mass density}.\footnote{Note that the quantity~$\Tr_V(P(0)) = \Tr_V(\nu(\hat{K}))$ coincides with the \emph{local particle density}~$f_{\textup{loc}}$ as introduced in~\cite[eq.~(4.4)]{continuum}. In order to avoid confusion, this notion will not be used in what follows.} In order to see that~$\Tr_V(P(0))$ can indeed be regarded as a density, let us assume that~$(\H, \F, d\rho)$ is a causal fermion system. Whenever~$P^{\varepsilon}(x,y)$ is a regularization of the kernel of the fermionic projector of the vacuum~$P(x,y)$ with regularization length~$\varepsilon$ (where~$P(x,y)$ coincides with~\eqref{(Dirac sea)}, cf.~\cite[eq.~(1.2.23)]{cfs}), its trace is given by (see~\cite[eq.~(2.5.1)]{cfs}) 
\begin{align*}
\Tr_{S_x} \big(P^{\varepsilon}(x,x) \big) \sim \frac{m}{\varepsilon^2} \qquad \text{for all~$x \in \supp d\rho$} \:. 
\end{align*}
Making use of the fact that the unit of mass equals one over length, we conclude that~$\Tr_{S_x} \big(P^{\varepsilon}(x,x) \big)$ is a density, which apparently is proportional to the mass~$m$. Carrying these observations over to~$\Tr_V (P(0))$ in the homogeneous case justifies the terminology of mass density.

A possible explanation for introducing the constraint~\eqref{(swf)} is that a similar side condition for the closed chain is imposed in the existence theorem~\cite[Theorem~6.1]{discrete}. Since the fermionic projector~$P(0) = \nu(\hat{K})$ can be diagonalized (up to an arbitrarily small error term) according to~\cite[Lemma~4.4]{continuum}, 
in order to develop the existence theory of minimizers in the homogeneous setting 
it seems promising to demand that constraint~\eqref{(swf)} is satisfied. On the other hand, following the original ideas in~\cite{pfp} and its modifications in~\cite{continuum}, it is natural to impose a boundedness constraint~\eqref{(boundedness constraint)}. The arguments in~\S \ref{S boundedness constraint} show that~\eqref{(boundedness constraint)} already implies condition~\eqref{(swf)}.

Let us finally discuss the remaining side condition~\eqref{(f)}. Since working with the spectral weight as appearing in the constraint~\eqref{(swf)} may be awkward, it might seem preferable to work with a similar condition which is more easy to handle. Bearing in mind that the operator~$\nu(\hat{K})$ may be diagonalized (up to an arbitrarily small error term) in virtue of~\cite[Lemma~4.4]{continuum} in such a way that its diagonal entries are ordered according to~\cite[eq.~(2.6)]{continuum}, the specific form of the signature matrix~$S$ (see~\eqref{(S)}) suggests to replace condition~\eqref{(swf)} by~\eqref{(f)},
\begin{align*}
\Tr_V(-S\nu(\hat{K})) = f \:.
\end{align*}
The same arguments as before illustrate that~$\Tr_V(-S\nu(\hat{K}))$ is a density; we refer to this quantity as \emph{particle density}.

\Thanks {{\em{Acknowledgments:}}
	I would like to thank Felix Finster for valuable comments on the manuscript. 
	I gratefully acknowledge financial support by the ``Studienstiftung des deutschen Volkes.'' 
		

\providecommand{\bysame}{\leavevmode\hbox to3em{\hrulefill}\thinspace}
\providecommand{\MR}{\relax\ifhmode\unskip\space\fi MR }
\providecommand{\MRhref}[2]{%
	\href{http://www.ams.org/mathscinet-getitem?mr=#1}{#2}
}

\providecommand{\href}[2]{#2}

\end{document}